\documentclass[pra,twocolumn]{revtex4-1}%
\usepackage{amsfonts}
\usepackage{amsmath}
\usepackage{amssymb}
\usepackage{graphicx}
\usepackage[colorlinks=true,linkcolor=blue,citecolor=red,plainpages=false,pdfpagelabels]%
{hyperref}%

\usepackage{listings}
\providecommand{\U}[1]{\protect\rule{.1in}{.1in}}
\newtheorem{theorem}{Theorem}

\newtheorem{corollary}{Corollary}

\newtheorem{definition}{Definition}

\newtheorem{proposition}{Proposition}
\newtheorem{remark}{Remark}

\newenvironment{proof}[1][Proof]{\noindent\textbf{#1.} }{\ \rule{0.5em}{0.5em}}
\begin{document}
\title{Entanglement cost and quantum channel simulation}
\author{Mark M. Wilde}
\affiliation{Hearne Institute for Theoretical Physics, Department of Physics and Astronomy,
Center for Computation and Technology, Louisiana State University, Baton
Rouge, Louisiana 70803, USA}

\begin{abstract}
This paper proposes a revised definition for the entanglement cost of a
quantum channel $\mathcal{N}$. In particular, it is defined here to be the
smallest rate at which entanglement is required, in addition to free classical
communication, in order to simulate $n$ calls to $\mathcal{N}$, such that the
most general discriminator cannot distinguish the $n$ calls to $\mathcal{N}$
from the simulation. The most general discriminator is one who tests the
channels in a sequential manner, one after the other, and this discriminator
is known as a quantum tester [Chiribella \textit{et al}., Phys.~Rev.~Lett.,
\textbf{101}, 060401 (2008)] or one who is implementing a quantum co-strategy
[Gutoski \textit{et al}., Symp.~Th.~Comp., 565 (2007)]. As such, the proposed
revised definition of entanglement cost of a quantum channel leads to a rate
that cannot be smaller than the previous notion of a channel's entanglement
cost [Berta \textit{et al}., IEEE Trans.~Inf.~Theory, \textbf{59}, 6779
(2013)], in which the discriminator is limited to distinguishing parallel
uses of the channel from the simulation. Under this revised notion, I prove
that the entanglement cost of certain teleportation-simulable channels is
equal to the entanglement cost of their underlying resource states. Then I
find single-letter formulas for the entanglement cost of some fundamental
channel models, including dephasing, erasure, three-dimensional
Werner--Holevo channels, epolarizing channels (complements of depolarizing channels), as well as single-mode pure-loss and pure-amplifier
bosonic Gaussian channels. These examples demonstrate that the resource theory of entanglement for quantum channels is not reversible. Finally, I discuss how to generalize the basic
notions to arbitrary resource theories.

\end{abstract}
\date{\today}
\startpage{1}
\endpage{10}
\maketitle

\section{Introduction}

The resource theory of entanglement \cite{BDSW96}\ has been one of the richest
contributions to quantum information theory \cite{H12,H06book,W17,Wat16}, and
these days, the seminal ideas coming from it are influencing diverse areas of
physics \cite{CG18}. A fundamental question in entanglement theory is to
determine the smallest rate at which Bell states (or ebits) are needed, along
with the assistance of free classical communication, in order to generate $n$
copies of an arbitrary bipartite state $\rho_{AB}$ reliably (in this
introduction, $n$ should be understood to be an arbitrarily large number)
\cite{BDSW96}. The optimal rate is known as the entanglement cost of
$\rho_{AB}$ \cite{BDSW96}, and a formal expression is known for this quantity
in terms of a regularization of the entanglement of formation \cite{HHT01}. An
upper bound in terms of entanglement of formation has been known for some time
\cite{BDSW96,HHT01}, while a lower bound in terms of a semi-definite programming quantity has been determined recently 
\cite{WD17}. Conversely, a related fundamental question is to determine the
largest rate at which one can distill ebits reliably from $n$ copies of
$\rho_{AB}$, again with the assistance of free classical communication
\cite{BDSW96}. This optimal rate is known as the distillable entanglement, and
various lower bounds \cite{DW05} and upper bounds \cite{Rai99,Rai01,CW04,WD16pra} are
known for~it.

The above resource theory is quite rich and interesting, but soon after
learning about it, one might immediately question its operational
significance. How are the bipartite states $\rho_{AB}$ established in the
first place? Of course, a quantum communication channel, such as a fiber-optic
or free-space link, is required. Consequently, in the same paper that
introduced the resource theory of entanglement \cite{BDSW96}, the authors
there appreciated the relevance of this point and proposed that the
distillation question could be extended to quantum channels. The distillation
question for channels is then as follows: given $n$ uses of a quantum channel
$\mathcal{N}_{A\rightarrow B}$ connecting a sender Alice to a receiver Bob,
along with the assistance of free classical communication, what is the optimal
rate at which these channels can produce ebits reliably \cite{BDSW96}? By invoking the
teleportation protocol \cite{BBC+93} and the fact that free classical
communication is allowed, this rate is also equal to the rate at which
arbitrary qubits can be reliably communicated by using the channel $n$ times
\cite{BDSW96}. The optimal rate is known as the distillable entanglement of
the channel \cite{BDSW96}, and various lower bounds \cite{DW05} and upper
bounds \cite{TGW14,TGW14b,W16,BW17} are now known for it, strongly related to
the bounds for distillable entanglement of states, as given above.

Some years after the distillable entanglement of a channel was proposed in
\cite{BDSW96}, the question converse to it was proposed and addressed in
\cite{BBCW13}. The authors of \cite{BBCW13} defined the entanglement cost of a
quantum channel $\mathcal{N}_{A\rightarrow B}$\ as the smallest rate at which
entanglement is required, in addition to the assistance of free classical
communication, in order to simulate $n$~uses of $\mathcal{N}_{A\rightarrow B}%
$. Key to their definition of entanglement cost is the particular notion of
simulation considered. In particular, the goal of their simulation protocol is
to simulate $n$~parallel uses of the channel, written as $(\mathcal{N}%
_{A\rightarrow B})^{\otimes n}$. Furthermore, they considered a simulation
protocol $\mathcal{P}_{A^{n}\rightarrow B^{n}}$ to have the following form:%
\begin{equation}
\mathcal{P}_{A^{n}\rightarrow B^{n}}(\omega_{A^{n}})\equiv\mathcal{L}%
_{A^{n}\overline{A}_{0}\overline{B}_{0}\rightarrow B^{n}}(\omega_{A^{n}%
}\otimes\Phi_{\overline{A}_{0}\overline{B}_{0}}),
\label{eq:berta-parallel-prot}%
\end{equation}
where $\omega_{A^{n}}$ is an arbitrary input state, $\mathcal{L}%
_{A^{n}\overline{A}_{0}\overline{B}_{0}\rightarrow B^{n}}$ is a free channel,
whose implementation is restricted to consist of local operations and
classical communication (LOCC) \cite{BDSW96,CLM+14}, and $\Phi_{\overline
{A}_{0}\overline{B}_{0}}$ is a maximally entangled resource state. For
$\varepsilon\in\left[  0,1\right]  $, the simulation is then considered
$\varepsilon$-distinguishable from $(\mathcal{N}_{A\rightarrow B})^{\otimes
n}$ if the following condition holds%
\begin{equation}
\frac{1}{2}\left\Vert (\mathcal{N}_{A\rightarrow B})^{\otimes n}%
-\mathcal{P}_{A^{n}\rightarrow B^{n}}\right\Vert _{\Diamond}\leq\varepsilon,
\label{eq:berta-sim}%
\end{equation}
where $\left\Vert \cdot\right\Vert _{\Diamond}$ denotes the diamond norm
\cite{Kit97}. The physical meaning of the above inequality is that it places a
limitation on how well any discriminator can distinguish the channel
$(\mathcal{N}_{A\rightarrow B})^{\otimes n}$ from the simulation
$\mathcal{P}_{A^{n}\rightarrow B^{n}}$ in a guessing game. Such a guessing
game consists of the discriminator preparing a quantum state $\rho_{RA^{n}}$,
the referee picking $(\mathcal{N}_{A\rightarrow B})^{\otimes n}$ or
$\mathcal{P}_{A^{n}\rightarrow B^{n}}$ at random and then applying it to the
$A^{n}$ systems of $\rho_{RA^{n}}$, and the discriminator finally performing a
quantum measurement on the systems $RB^{n}$. If the inequality in
\eqref{eq:berta-sim} holds, then the probability that the discriminator can
correctly distinguish the channel from its simulation is bounded from above by
$\frac{1}{2}\left(  1+\varepsilon\right)  $, regardless of the particular
state $\rho_{RA^{n}}$ and final measurement chosen for his distinguishing
strategy \cite{Kit97,H69,H73,Hel76}. Thus, if $\varepsilon$ is close to zero,
then this probability is not much better than random guessing, and in this
case, the channels are considered nearly indistinguishable and the simulation
thus reliable.

In parallel to the above developments in entanglement theory, there have
indubitably been many advances in the theory of quantum channel discrimination
\cite{CDP08b,CDP08a,GDP09,DFY09,Harrow10,CMW16} and related developments in
the theory of quantum interactive proof systems \cite{GW07,G09,G12,GAJ18}.
Notably, the most general method for distinguishing a quantum memory channel
from another one consists of a quantum-memory-assisted discrimination protocol
\cite{CDP08a,GDP09}. In the language of quantum interactive proof systems,
memory channels are called strategies and memory-assisted discrimination
protocols are called co-strategies \cite{GW07,G09,G12}. For a visual
illustration of the physical setup, please consult \cite[Figure~2]{CDP08a} or
\cite[Figure~2]{GW07}. In subsequent work after \cite{GW07,CDP08a}, a number
of theoretical results listed above have been derived related to memory
channel discrimination or quantum strategies.

\begin{figure*}[ptb]
\begin{center}
\includegraphics[
width=4.8399in
]
{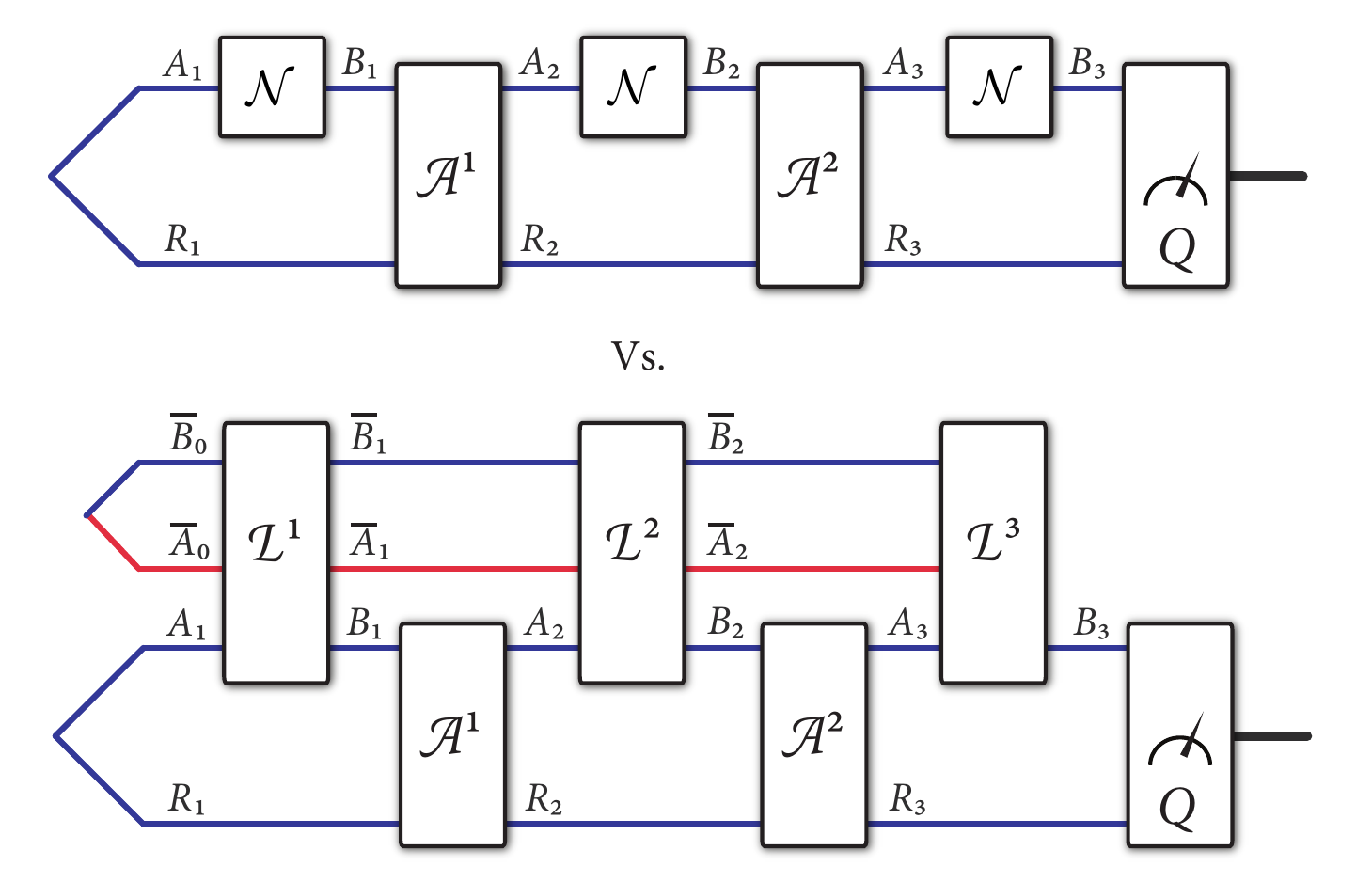}
\end{center}
\caption{The top part of the figure displays a three-round interaction between
the discriminator and the simulator in the case that the actual channel
$\mathcal{N}_{A\rightarrow B}$ is called three times. The bottom part of the
figure displays the interaction between the discriminator and the simulator in
the case that the simulation of three channel uses is called.}%
\label{fig:adaptive-prot}%
\end{figure*}
The aforementioned developments in the theory of quantum channel
discrimination indicate that the notion of channel simulation proposed in
\cite{BBCW13} is not the most general notion that could be considered. In
particular, if a simulator is claiming to have simulated $n$ uses of the
channel $\mathcal{N}_{A\rightarrow B}$, then the discriminator should be able
to test this assertion in the most general way possible, as given in
\cite{GW07,CDP08a,GDP09}. That is, we would like for the simulation to pass
the strongest possible test that could be performed to distinguish it from the
$n$ uses of $\mathcal{N}_{A\rightarrow B}$. Such a test allows for the
discriminator to prepare an arbitary state $\rho_{R_{1}A_{1}}$, call the first
channel use $\mathcal{N}_{A_{1}\rightarrow B_{1}}$ or its simulation, apply an
arbitrary channel $\mathcal{A}_{R_{1}B_{1}\rightarrow R_{2}A_{2}}^{(1)}$, call
the second channel use or its simulation, etc. After the $n$th call is made,
the discriminator then performs a joint measurement on the remaining quantum
systems. See Figure~\ref{fig:adaptive-prot} for a visual depiction. If the
simulation is good, then the probability for the discriminator to distinguish
the $n$ channels from the simulation should be no larger than $\frac{1}%
{2}\left(  1+\varepsilon\right)  $, for small $\varepsilon$.

In this paper, I propose a new definition for the entanglement cost of a
channel $\mathcal{N}_{A\rightarrow B}$, such that it is the smallest rate at
which ebits are needed, along with the assistance of free classical
communication, in order to simulate $n$ uses of $\mathcal{N}_{A\rightarrow B}%
$, in such a way that a discriminator performing the most stringest test, as
described above, cannot distinguish the simulation from $n$ actual calls of
$\mathcal{N}_{A\rightarrow B}$ (Section~\ref{sec:ent-cost-new-def}). Here I
denote the optimal rate by $E_{C}(\mathcal{N})$, and the prior quantity
defined in \cite{BBCW13} by $E_{C}^{(p)}(\mathcal{N})$, given that the
simulation there was only required to pass a less stringent parallel
discrimination test, as discussed above. Due to the fact that it is more
difficult to pass the simulation test as specified by the new definition, it follows that
$E_{C}(\mathcal{N})\geq E_{C}^{(p)}(\mathcal{N})$ (discussed in more detail in
what follows). After establishing definitions, I then prove a general
upper bound on the entanglement cost of a quantum channel, using the notion of
teleportation simulation (Section~\ref{sec:upper-bnd-TP-sim}). I prove that
the entanglement cost of certain ``resource-seizable,''
teleportation-simulable channels takes on a particularly simple form
(Section~\ref{sec:eq-res-seize}), which allows for concluding single-letter
formulas for the entanglement cost of dephasing, erasure, 
three-dimensional Werner--Holevo channels, and epolarizing channels (complements of depolarizing channels), as detailed in Section~\ref{sec:examples}. Note
that the result about entanglement cost of dephasing channels solves an open
question from \cite{BBCW13}. I then extend the results to the case of bosonic
Gaussian channels (Section~\ref{sec:BGCs}), proving single-letter formulas for
the entanglement cost of fundamental channel models, including pure-loss and
pure-amplifier channels (Theorem~\ref{thm:pure-bosonic-formulas} in
Section~\ref{sec:pure-BGCs}). These examples lead to the conclusion that the resource theory of entanglement for quantum channels is not reversible. I also prove that the entanglement cost
of thermal, amplifier,
and additive-noise bosonic Gaussian channels is bounded from below by the entanglement cost 
of their
``Choi states.''
In Section~\ref{sec:resource-theory-gen}, I discuss how to generalize the basic notions to other resource theories.
Finally, Section~\ref{sec:concl}
concludes with a summary and some open questions.

\section{Notions of quantum channel simulation}

In this section, I review the definition of entanglement cost of a quantum
channel, as detailed in \cite{BBCW13}, and I also review the main theorem from
\cite{BBCW13}. After that, I propose the revised definition of a channel's
entanglement cost.

Before starting, let us define a maximally entangled state $\Phi_{AB}$ of
Schmidt rank $d$ as%
\begin{equation}
\Phi_{AB}\equiv\frac{1}{d}\sum_{i,j=1}^d|i\rangle\langle j|_{A}\otimes
|i\rangle\langle j|_{B},
\end{equation}
where $\{|i\rangle_{A}\}_{i}$ and $\{|i\rangle_{B}\}_{i}$ are orthonormal
bases. An LOCC channel $\mathcal{L}_{A^{\prime}B^{\prime}\rightarrow AB}$ is a
bipartite channel that can be written in the following form:%
\begin{equation}
\mathcal{L}_{A^{\prime}B^{\prime}\rightarrow AB}=\sum_{y}\mathcal{E}%
_{A^{\prime}\rightarrow A}^{y}\otimes\mathcal{F}_{B^{\prime}\rightarrow B}%
^{y}, \label{eq:LOCC-channel-decomp}%
\end{equation}
where $\{\mathcal{E}_{A^{\prime}\rightarrow A}^{y}\}_{y}$ and $\{\mathcal{F}%
_{B^{\prime}\rightarrow B}^{y}\}_{y}$ are sets of completely positive,
trace-non-increasing maps, such that the sum map $\sum_{y}\mathcal{E}%
_{A^{\prime}\rightarrow A}^{y}\otimes\mathcal{F}_{B^{\prime}\rightarrow B}%
^{y}$ is a quantum channel (completely positive and trace preserving)
\cite{CLM+14}. However, not every  channel of the form in \eqref{eq:LOCC-channel-decomp} is an LOCC channel (there are separable channels of the form in \eqref{eq:LOCC-channel-decomp} that are not implementable by LOCC \cite{PhysRevA.59.1070}). The diamond norm of the difference of two channels
$\mathcal{R}_{A\rightarrow B}$ and $\mathcal{S}_{A\rightarrow B}$ is defined
as \cite{Kit97}
\begin{equation}
\left\Vert \mathcal{R}-\mathcal{S}\right\Vert _{\Diamond}\equiv\sup_{\psi
_{RA}}\left\Vert \mathcal{R}_{A\rightarrow B}(\psi_{RA})-\mathcal{S}%
_{A\rightarrow B}(\psi_{RA})\right\Vert _{1},
\end{equation}
where the optimization is with respect to all pure bipartite states $\psi
_{RA}$ with system $R$ isomorphic to system~$A$ and the trace norm of an operator $X$ is defined as $\| X\|_1 \equiv \operatorname{Tr}\{\sqrt{X^\dag X}\}$. The operational
interpretation of the diamond norm is that it is related to the maximum
success probability $p_{\text{succ}}(\mathcal{R},\mathcal{S})$\ for any
physical experiment, of the kind discussed after \eqref{eq:berta-sim}, to
distinguish the channels $\mathcal{R}$ and $\mathcal{S}$:
\begin{equation}
p_{\text{succ}}(\mathcal{R},\mathcal{S})=\frac{1}{2}\left(  1+\frac{1}%
{2}\left\Vert \mathcal{R}-\mathcal{S}\right\Vert _{\Diamond}\right)  .
\end{equation}

\subsection{Entanglement cost of a quantum channel from~\cite{BBCW13}}

\label{sec:ent-cost-berta}Let us now review the notion of entanglement cost
from \cite{BBCW13}. Fix $n,M\in\mathbb{N}$, $\varepsilon\in\left[  0,1\right]
$, and a quantum channel $\mathcal{N}_{A\rightarrow B}$. According to
\cite{BBCW13}, an $(n,M,\varepsilon)$ (parallel)\ LOCC-assisted channel
simulation code consists of an LOCC\ channel $\mathcal{L}_{A^{n}\overline
{A}_{0}\overline{B}_{0}\rightarrow B^{n}}$ and a maximally entangled resource
state $\Phi_{\overline{A}_{0}\overline{B}_{0}}$\ of Schmidt rank $M$, such
that together they implement a simulation channel $\mathcal{P}_{A^{n}%
\rightarrow B^{n}}$, as defined in \eqref{eq:berta-parallel-prot}. In this
model, to be clear, we assume that Alice has access to all systems labeled by
$A$, Bob has access to all systems labeled by $B$, and they are in distant
laboratories. The simulation $\mathcal{P}_{A^{n}\rightarrow B^{n}}$ is
considered $\varepsilon$-distinguishable from $n$ parallel calls
$(\mathcal{N}_{A\rightarrow B})^{\otimes n}$ of the actual channel
$\mathcal{N}_{A\rightarrow B}$ if the condition in
\eqref{eq:berta-sim}\ holds. Note here again that the condition in
\eqref{eq:berta-sim} corresponds to a discriminator who is restricted to
performing only a parallel test to distinguish the $n$ calls of $\mathcal{N}%
_{A\rightarrow B}$ from its simulation. Let us also note here that the
condition in \eqref{eq:berta-sim} can be understood as the simulation
$\mathcal{P}_{A^{n}\rightarrow B^{n}}$ providing an approximate teleportation
simulation of $(\mathcal{N}_{A\rightarrow B})^{\otimes n}$, in the language of
the later work of \cite{KW17a}.

A rate $R$ is said to be achievable for (parallel)\ channel simulation of
$\mathcal{N}_{A\rightarrow B}$\ if for all $\varepsilon\in(0,1]$, $\delta>0$,
and sufficiently large $n$, there exists an $(n,2^{n\left[  R+\delta\right]
},\varepsilon)$ LOCC-assisted channel simulation code. The
(parallel)\ entanglement cost $E_{C}^{(p)}(\mathcal{N})$ of the channel
$\mathcal{N}$ is equal to the infimum of all achievable rates, with the
superscript $(p)$ indicating that the test of the simulation is restricted to
be a parallel discrimination test.

The main result of \cite{BBCW13} is that the channel's entanglement cost
$E_{C}^{(p)}(\mathcal{N})$ is equal to the regularization of its entanglement
of formation. To state this result precisely, recall that the entanglement of
formation of a bipartite state $\rho_{AB}$ is defined as \cite{BDSW96}%
\begin{multline}
E_{F}(A;B)_{\rho}\equiv\label{eq:EoF}\\
\inf\left\{  \sum_{x}p_{X}(x)H(A)_{\psi^{x}}:\rho_{AB}=\sum_{x}p_{X}( x)
\psi_{AB}^{x}\right\}  ,
\end{multline}
where the infimum is with respect to all convex decompositions of $\rho_{AB}$
into pure states $\psi_{AB}^{x}$ and%
\begin{equation}
H(A)_{\psi^{x}}\equiv-\operatorname{Tr}\{\psi_{A}^{x}\log_{2}\psi_{A}^{x}\}
\end{equation}
is the quantum entropy of the marginal state $\psi_{A}^{x}=\operatorname{Tr}%
_{B}\{\psi_{AB}^{x}\}$. The entanglement of formation does not increase under
the action of an LOCC\ channel \cite{BDSW96}. A channel's entanglement of
formation $E_{F}(\mathcal{N})$ is then defined as%
\begin{equation}
E_{F}(\mathcal{N})\equiv\sup_{\psi_{RA}}E_{F}(R;B)_{\omega},
\end{equation}
where $\omega_{RB}\equiv\mathcal{N}_{A\rightarrow B}(\psi_{RA})$, and it
suffices to take the optimization with respect to a pure state input
$\psi_{RA}$, with system $R$ isomorphic to system $A$, due to purification, the
Schmidt decomposition theorem, and the LOCC\ monotonicity of entanglement of
formation \cite{BDSW96}. We can now state the main result of \cite{BBCW13}
described above:%
\begin{equation}
E_{C}^{(p)}(\mathcal{N})=\lim_{n\rightarrow\infty}\frac{1}{n}E_{F}%
(\mathcal{N}^{\otimes n}). \label{eq:main-result-berta}%
\end{equation}
The regularized formula on the right-hand side may be difficult to evaluate in
general, and thus can only be considered a formal expression, but if the
additivity relation $\frac{1}{n}E_{F}(\mathcal{N}^{\otimes n}) =
E_{F}(\mathcal{N})$ holds for a given channel $\mathcal{N}$ for all $n \geq 1$, then it
simplifies significantly as $E_{C}^{(p)}(\mathcal{N}) = E_{F}(\mathcal{N})$.

\subsection{Proposal for a revised notion of entanglement cost of a channel}

\label{sec:ent-cost-new-def}

Now I propose the new or revised definition for entanglement cost of a
channel. As motivated in the introduction, a parallel test of channel
simulation is not the most general kind of test that can be considered.
Thus, the new definition proposes that the entanglement cost of a channel
should incorporate the most stringent test possible.

To begin with, let us fix $n,M\in\mathbb{N}$, $\varepsilon\in\left[
0,1\right]  $, and a quantum channel $\mathcal{N}_{A\rightarrow B}$. We define
an $(n,M,\varepsilon)$ (sequential)\ LOCC-assisted channel simulation code to
consist of a maximally entangled resource state $\Phi_{\overline{A}%
_{0}\overline{B}_{0}}$\ of Schmidt rank $M$ and a set%
\begin{equation}
\{\mathcal{L}_{A_{i}\overline{A}_{i-1}\overline{B}_{i-1}\rightarrow
B_{i}\overline{A}_{i}\overline{B}_{i}}^{(i)}\}_{i=1}^{n} \label{eq:sim-prot}%
\end{equation}
of LOCC channels. Note that the systems $\overline{A}_{n}\overline{B}_{n}$ of
the final LOCC\ channel $\mathcal{L}_{A_{n}\overline{A}_{n-1}\overline
{B}_{n-1}\rightarrow B_{n}\overline{A}_{n}\overline{B}_{n}}^{(n)}$ can be
taken trivial without loss of generality. As before, Alice has access to all
systems labeled by $A$, Bob has access to all systems labeled by $B$, and they
are in distant laboratories. The structure of this simulation protocol is
intended to be compatible with a discrimination strategy that can test the
actual $n$ channels versus the above simulation in a sequential way, along the
lines discussed in \cite{CDP08a,GDP09}\ and\ \cite{G12}. I later show how this
encompasses the parallel tests discussed in the previous section.

A sequential discrimination strategy consists of an initial state
$\rho_{R_{1}A_{1}}$, a set $\{\mathcal{A}_{R_{i}B_{i}\rightarrow
R_{i+1}A_{i+1}}^{(i)}\}_{i=1}^{n-1}$ of adaptive channels, and a quantum
measurement $\{Q_{R_{n}B_{n}},I_{R_{n}B_{n}}-Q_{R_{n}B_{n}}\}$. Let us employ
the shorthand $\{\rho,\mathcal{A},Q\}$ to abbreviate such a discrimination
strategy. Note that, in performing a discrimination strategy, the
discriminator has a full description of the channel $\mathcal{N}_{A\rightarrow
B}$ and the simulation protocol, which consists of $\Phi_{\overline{A}%
_{0}\overline{B}_{0}}$ and the set in \eqref{eq:sim-prot}. If this
discrimination strategy is performed on the $n$ uses of the actual channel
$\mathcal{N}_{A\rightarrow B}$, the relevant states involved are%
\begin{equation}
\rho_{R_{i+1}A_{i+1}}\equiv\mathcal{A}_{R_{i}B_{i}\rightarrow R_{i+1}A_{i+1}%
}^{(i)}(\rho_{R_{i}B_{i}}),
\end{equation}
for $i\in\left\{  1,\ldots,n-1\right\}  $ and
\begin{equation}
\rho_{R_{i}B_{i}}\equiv\mathcal{N}_{A_{i}\rightarrow B_{i}}(\rho_{R_{i}A_{i}%
}),
\end{equation}
for $i\in\left\{  1,\ldots,n\right\}  $. If this discrimination strategy is
performed on the simulation protocol discussed above, then the relevant states
involved are%
\begin{align}
\tau_{R_{1}B_{1}\overline{A}_{1}\overline{B}_{1}}  &  \equiv\mathcal{L}%
_{A_{1}\overline{A}_{0}\overline{B}_{0}\rightarrow B_{1}\overline{A}%
_{1}\overline{B}_{1}}^{(1)}(\tau_{R_{1}A_{1}}\otimes\Phi_{\overline{A}%
_{0}\overline{B}_{0}}),\nonumber\\
\tau_{R_{i+1}A_{i+1}\overline{A}_{i}\overline{B}_{i}}  &  \equiv
\mathcal{A}_{R_{i}B_{i}\rightarrow R_{i+1}A_{i+1}}^{(i)}(\tau_{R_{i}%
B_{i}\overline{A}_{i}\overline{B}_{i}}),
\end{align}
for $i\in\left\{  1,\ldots,n-1\right\}  $, where $\tau_{R_{1}A_{1}}%
=\rho_{R_{1}A_{1}}$, and
\begin{equation}
\tau_{R_{i}B_{i}\overline{A}_{i}\overline{B}_{i}}\equiv\mathcal{L}%
_{A_{i}\overline{A}_{i-1}\overline{B}_{i-1}\rightarrow B_{i}\overline{A}%
_{i}\overline{B}_{i}}^{(i)}(\tau_{R_{i}A_{i}\overline{A}_{i-1}\overline
{B}_{i-1}}),
\end{equation}
for $i\in\left\{  2,\ldots,n\right\}  $. The discriminator then performs the
measurement $\{Q_{R_{n}B_{n}},I_{R_{n}B_{n}}-Q_{R_{n}B_{n}}\}$ and guesses
\textquotedblleft actual channel\textquotedblright\ if the outcome is
$Q_{R_{n}B_{n}}$ and \textquotedblleft simulation\textquotedblright\ if the
outcome is $I_{R_{n}B_{n}}-Q_{R_{n}B_{n}}$. Figure~\ref{fig:adaptive-prot}%
\ depicts the discrimination strategy in the case that the actual channel is
called $n=3$ times and in the case that the simulation is performed.

If the \textit{a priori} probabilities for the actual channel or simulation
are equal, then the success probability of the discriminator in distinguishing
the channels is given by%
\begin{multline}
\frac{1}{2}\left[  \operatorname{Tr}\{Q_{R_{n}B_{n}}\rho_{R_{n}B_{n}%
}\}+\operatorname{Tr}\{\left(  I_{R_{n}B_{n}}-Q_{R_{n}B_{n}}\right)
\tau_{R_{n}B_{n}}\}\right] \\
\leq\frac{1}{2}\left(  1+\frac{1}{2}\left\Vert \rho_{R_{n}B_{n}}-\tau
_{R_{n}B_{n}}\right\Vert _{1}\right)  ,
\end{multline}
where the latter inequality is well known from the theory of quantum state
discrimination \cite{H69,H73,Hel76}. For this reason, we say that the $n$
calls to the actual channel $\mathcal{N}_{A\rightarrow B}$ are $\varepsilon
$-distinguishable from the simulation if the following condition holds for the
respective final states%
\begin{equation}
\frac{1}{2}\left\Vert \rho_{R_{n}B_{n}}-\tau_{R_{n}B_{n}}\right\Vert _{1}%
\leq\varepsilon.
\end{equation}
If this condition holds for all possible discrimination strategies
$\{\rho,\mathcal{A},Q\}$, i.e., if%
\begin{equation}
\frac{1}{2}\sup_{\left\{  \rho,\mathcal{A}\right\}  }\left\Vert \rho
_{R_{n}B_{n}}-\tau_{R_{n}B_{n}}\right\Vert _{1}\leq\varepsilon,
\label{eq:good-sim}%
\end{equation}
then the simulation protocol constitutes an $(n,M,\varepsilon)$ channel
simulation code. It is worthwhile to remark: If we ascribe the shorthand
$(\mathcal{N})^{n}$ for the $n$ uses of the channel and the shorthand
$(\mathcal{L})^{n}$ for the simulation, then the condition in
\eqref{eq:good-sim} can be understood in terms of the $n$-round strategy norm
of \cite{CDP08a,GDP09,G12}:%
\begin{equation}
\frac{1}{2}\left\Vert (\mathcal{N})^{n}-(\mathcal{L})^{n}\right\Vert
_{\Diamond,n}\leq\varepsilon.
\end{equation}

As before, a rate $R$ is achievable for (sequential)\ channel simulation of
$\mathcal{N}$ if for all $\varepsilon\in(0,1]$, $\delta>0$, and sufficiently
large $n$, there exists an $(n,2^{n\left[  R+\delta\right]  },\varepsilon)$
(sequential)\ channel simulation code for $\mathcal{N}$. We define the
(sequential)\ entanglement cost $E_{C}(\mathcal{N})$ of the channel
$\mathcal{N}$ to be the infimum of all achievable rates. Due to the fact that
this notion is more general, we sometimes simply refer to $E_{C}(\mathcal{N})$
as the entanglement cost of the channel $\mathcal{N}$ in what follows.

\subsection{LOCC monotonicity of the entanglement cost}

\label{sec:LOCC-mon-cost}

Let us note here that if a channel $\mathcal{N}_{A\rightarrow B}$ can be
realized from another channel $\mathcal{M}_{A^{\prime}\rightarrow B^{\prime}}$
via a preprocessing LOCC\ channel $\mathcal{L}_{A\rightarrow A^{\prime}%
A_{M}B_{M}}^{\text{pre}}$ and a postprocessing LOCC\ channel $\mathcal{L}%
_{B^{\prime}A_{M}B_{M}\rightarrow B}^{\text{post}}$\ as%
\begin{equation}
\mathcal{N}_{A\rightarrow B}=\mathcal{L}_{B^{\prime}A_{M}B_{M}\rightarrow
B}^{\text{post}}\circ\mathcal{M}_{A^{\prime}\rightarrow B^{\prime}}%
\circ\mathcal{L}_{A\rightarrow A^{\prime}A_{M}B_{M}}^{\text{pre}},
\label{eq:pre-post}%
\end{equation}
then it follows that any $(n,M,\varepsilon)$ protocol for sequential channel
simulation of $\mathcal{M}_{A^{\prime}\rightarrow B^{\prime}}$ realizes an
$(n,M,\varepsilon)$ protocol for sequential channel simulation of
$\mathcal{N}_{A\rightarrow B}$. This is  an immediate consequence of
the fact that the best strategy for discriminating $\mathcal{N}_{A\rightarrow
B}$ from its simulation can be understood as a particular 
strategy for discriminating $\mathcal{M}_{A^{\prime}\rightarrow B^{\prime}}$ from a simulation of $\mathcal{M}_{A^{\prime}\rightarrow B^{\prime}}$, due to the
structural decomposition in \eqref{eq:pre-post}. Following definitions, a
simple consequence is the following LOCC\ monotonicity inequality for the
entanglement cost of these channels:%
\begin{equation}
E_{C}(\mathcal{N})\leq E_{C}(\mathcal{M}).
\label{eq:DP-bnd-pre-post}
\end{equation}
Thus, it takes more or the same entanglement to simulate the channel
$\mathcal{M}$ than it does to simulate $\mathcal{N}$. Furthermore, the 
decomposition in \eqref{eq:pre-post} and the bound in \eqref{eq:DP-bnd-pre-post} can be used to bound the entanglement cost of a
channel $\mathcal{M}$ from below. Note that the structure in
\eqref{eq:pre-post} was discussed recently in the context of general resource
theories \cite[Section~III-D-5]{CG18}.

\subsection{Parallel tests as a special case of sequential tests}

A parallel test of the form described in Section~\ref{sec:ent-cost-berta}\ is
a special case of the sequential test outlined above. One can see this in two
seemingly different ways. First, we can think of the sequential strategy
taking a particular form. The state $\xi_{RA_{1}A_{2}\cdots A_{n}}$ is
prepared, and here we identify systems $RA_{2}\cdots A_{n}$ with system
$R_{1}$ of $\rho_{R_{1}A_{1}}$ in an adaptive protocol and system $A_{1}$ of
$\xi_{RA_{1}A_{2}\cdots A_{n}}$ with system $A_{1}$ of $\rho_{R_{1}A_{1}}$.
Then the channel $\mathcal{N}_{A_{1}\rightarrow B_{1}}$ or its simulation is
called. After that, the action of the first adaptive channel is simply to swap
in system $A_{2}$ of $\xi_{RA_{1}A_{2}\cdots A_{n}}$ to the second call of the
channel $\mathcal{N}_{A_{2}\rightarrow B_{2}}$ or its simulation, while
keeping systems $RB_{1}A_{3}\cdots A_{n}$ as part of the reference $R_{2}$ of
the state $\rho_{R_{2}A_{2}}$. Then this iterates and the final measurement is
performed on all of the remaining systems.

The other way to see how a parallel test is a special kind of sequential test
is to rearrange the simulation protocol as has been done in
Figure~\ref{fig:sim-parallel}. Here, we see that the simulation protocol has a
memory structure, and it is clear that the simulation protocol can accept as
input a state $\xi_{RA_{1}A_{2}\cdots A_{n}}$ and outputs a state on systems
$RB_{1}\cdots B_{n}$, which can subsequently be measured.

\begin{figure}[ptb]
\begin{center}
\includegraphics[
width=3.3399in
]
{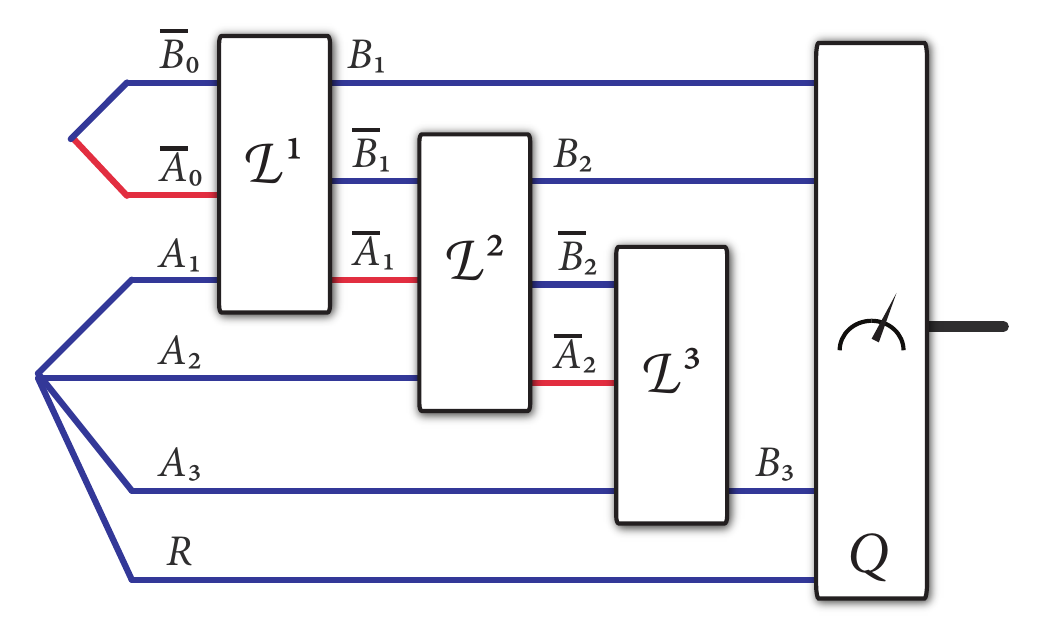}
\end{center}
\caption{The simulation protocol from the bottom part of
Figure~\ref{fig:adaptive-prot} rewritten to clarify that it can participate in
a parallel channel simulation test.}%
\label{fig:sim-parallel}%
\end{figure}

As a consequence of this reduction, any $(n,M,\varepsilon)$ sequential channel
simulation protocol can serve as an $(n,M,\varepsilon)$ parallel channel
simulation protocol. Furthermore, if $R$ is an achievable rate for sequential
channel simulation, then it is also an achievable rate for parallel channel
simulation. Finally, these reductions imply the following inequality:%
\begin{equation}
E_{C}(\mathcal{N})\geq E_{C}^{(p)}(\mathcal{N}).
\label{eq:sequential-greater-than-para}%
\end{equation}
Intuitively, one might sometimes require more entanglement in order to pass
the more stringest test that occurs in sequential channel simulation. As a
consequence of \eqref{eq:main-result-berta} and
\eqref{eq:sequential-greater-than-para}, we have that%
\begin{equation}
E_{C}(\mathcal{N})\geq\lim_{n\rightarrow\infty}\frac{1}{n}E_{F}(\mathcal{N}%
^{\otimes n}).
\end{equation}
It is an interesting question (not addressed here) to determine if there
exists a channel such that the inequality in
\eqref{eq:sequential-greater-than-para} is strict.

If desired, it is certainly possible to obtain a non-asymptotic, weak-converse
bound that implies the above bound after taking limits. Let us state this
bound as follows:

\begin{proposition}
\label{prop:non-asym-lower-bnd-e-cost}Let $\mathcal{N}_{A\rightarrow B}$ be a
quantum channel, and let $n,M\in\mathbb{N}$ and $\varepsilon\in\left[
0,1\right]  $. Set $d=\min\left\{  \left\vert A\right\vert ,\left\vert
B\right\vert \right\}  $, i.e., the minimum of the input and output dimensions
of the channel $\mathcal{N}_{A\rightarrow B}$. Then the following bound holds
for any $(n,M,\varepsilon)$ sequential channel simulation code:%
\begin{equation}
\frac{1}{n}\log_{2}M\geq\frac{1}{n}E_{F}(\mathcal{N}^{\otimes n}%
)-\sqrt{\varepsilon}\log d-\frac{1}{n}g_{2}(\sqrt{\varepsilon}),
\end{equation}
where $\frac{1}{n}\log_{2}M$ is understood as the non-asymptotic entanglement
cost of the protocol and the bosonic entropy function $g_{2}(x)$ is defined
for $x\geq0$ as%
\begin{equation}
g_{2}(x)\equiv\left(  x+1\right)  \log_{2}(x+1)-x\log_{2}x.
\label{eq:bosonic-ent}%
\end{equation}

\end{proposition}

\begin{proof}
To see this, suppose that there exists an $(n,M,\varepsilon)$ protocol for
sequential channel simulation. Then by the above reasoning (also see
Figure~\ref{fig:sim-parallel}), it can be thought of as a parallel channel
simulation protocol, such that the criterion in \eqref{eq:berta-sim}\ holds.
Suppose that $\psi_{RA_{1}\cdots A_{n}}$ is a test input state, with $|R| =
|A|^{n}$, leading to $\omega_{RB_{1}\cdots B_{n}}=(\mathcal{N}_{A\rightarrow
B})^{\otimes n}(\psi_{RA_{1}\cdots A_{n}})$ when the actual channels are
applied and $\sigma_{RB_{1}\cdots B_{n}}$ when the simulation is applied. Then
we have that%
\begin{align}
&  E_{F}(R;B_{1}\cdots B_{n})_{\omega}\nonumber\\
&  \leq E_{F}(R;B_{1}\cdots B_{n})_{\sigma}+n\sqrt{\varepsilon}\log
d+g_{2}(\sqrt{\varepsilon})\nonumber\\
&  \leq E_{F}(RA_{1}\cdots A_{n}\overline{A}_{0};\overline{B}_{0}%
)_{\psi\otimes\Phi}+n\sqrt{\varepsilon}\log d+g_{2}(\sqrt{\varepsilon
})\nonumber\\
&  =E_{F}(\overline{A}_{0};\overline{B}_{0})_{\Phi}+n\sqrt{\varepsilon}\log
d+g_{2}(\sqrt{\varepsilon})\nonumber\\
&  =\log_{2}M+n\sqrt{\varepsilon}\log d+g_{2}(\sqrt{\varepsilon}).
\end{align}
The first inequality follows from the condition in \eqref{eq:good-sim}, as
well as from the continuity bound for entanglement of formation from
\cite[Corollary~4]{Winter15}. The second inequality follows from the
LOCC\ monotonicity of the entanglement of formation \cite{BDSW96}, here
thinking of the person who possesses systems $RA_{1}\cdots A_{n}$ to be in the
same laboratory as the one possessing the systems $\overline{A}_{i}$, while
the person who possesses the $\overline{B}_{i}$ systems is in a different
laboratory. The first equality follows from the fact that $\psi_{RA_{1}\cdots
A_{n}}$ is in tensor product with $\Phi_{\overline{A}_{0}\overline{B}_{0}}$,
so that by a local channel, one may remove $\psi_{RA_{1}\cdots A_{n}}$ or
append it for free. The final equality follows because the entanglement of
formation of the maximally entangled state is equal to the logarithm of its
Schmidt rank. Since the bound holds uniformly regardless of the input state
$\psi_{RA_{1}\cdots A_{n}}$, after an optimization and a rearrangement we
conclude the stated lower bound on the non-asymptotic entanglement cost
$\frac{1}{n}\log_{2}M$ of the protocol.
\end{proof}

\begin{remark}
\label{rem:ent-break}Let us note here that the entanglement cost of a quantum
channel is equal to zero if and only if the channel is entanglement-breaking
\cite{HSR03,Holevo2008}. The \textquotedblleft if-part\textquotedblright%
\ follows as a straightforward consequence of definitions and the fact that
these channels can be implemented as a measurement followed by a preparation
\cite{HSR03,Holevo2008}, given that this measure-prepare procedure is a
particular kind of LOCC\ and thus allowed for free (without any cost)\ in the
above model. The \textquotedblleft only-if\textquotedblright\ part follows
from \eqref{eq:sequential-greater-than-para}\ and \cite[Corollary~18]{BBCW13},
the latter of which depends on the result from \cite{YHHS05}.
\end{remark}

\section{Bounds for the entanglement cost of teleportation-simulable channels}

\subsection{Upper bound on the entanglement cost of teleportation-simulable
channels}

\label{sec:upper-bnd-TP-sim}

The most trivial method for simulating a channel is to employ the
teleportation protocol \cite{BBC+93}\ directly. In this method, Alice and Bob
could use the teleportation protocol so that Alice could transmit the input of
the channel to Bob, who could then apply the channel. Repeating this $n$
times, this trivial method would implement an $(n,\left\vert A\right\vert
^{n},0)$ simulation protocol in either the parallel or sequential model.
Alternatively, Alice could apply the channel first and then teleport the
output to Bob, and repeating this $n$ times would implement an $(n,\left\vert
B\right\vert ^{n},0)$ simulation protocol in either the parallel or sequential
model. Thus, they could always achieve a rate of $\log_{2}(\min\left\{
\left\vert A\right\vert ,\left\vert B\right\vert \right\}  )$ using this
approach, and this reasoning establishes a simple dimension upper bound on the
entanglement cost of a channel:%
\begin{equation}
E_{C}(\mathcal{N}_{A\rightarrow B})\leq\log_{2}(\min\left\{  \left\vert
A\right\vert ,\left\vert B\right\vert \right\}  ).
\end{equation}
In this context, also see \cite[Proposition~9]{KW17a}.

A less trivial approach is to exploit the fact that some channels of interest
could be teleportation-simulable with associated resource state $\omega
_{A^{\prime}B^{\prime}}$, in which the resource state need not be a maximally
entangled state (see \cite[Section~V]{BDSW96} and \cite[Eq.~(11)]{HHH99}).
Recall from these references that a channel $\mathcal{N}_{A\rightarrow B}$ is
teleportation-simulable with associated resource state $\omega_{A^{\prime
}B^{\prime}}$ if there exists an LOCC channel $\mathcal{L}_{AA^{\prime
}B^{\prime}\rightarrow B}$ such that the following equality holds for all
input states $\rho_{A}$:%
\begin{equation}
\mathcal{N}_{A\rightarrow B}(\rho_{A})=\mathcal{L}_{AA^{\prime}B^{\prime
}\rightarrow B}(\rho_{A}\otimes\omega_{A^{\prime}B^{\prime}}).
\label{eq:tp-sim}%
\end{equation}
If a channel possesses this structure, then we arrive at the following upper
bound on the entanglement cost:

\begin{proposition}
\label{prop:ach-TP-sim}Let $\mathcal{N}_{A\rightarrow B}$ be a quantum channel
that is teleportation-simulable with associated resource state $\omega
_{A^{\prime}B^{\prime}}$, as defined in \eqref{eq:tp-sim}. Let $n,M\in
\mathbb{N}$ and $\varepsilon\in\left(  0,1\right)  $. Then there exists an
$(n,M,\sqrt{\varepsilon})$ sequential channel simulation code satisfying the
following bound%
\begin{equation}
\frac{1}{n}\log_{2}M\leq\frac{1}{n}E_{F,0}^{\varepsilon/2}(A^{\prime
n};B^{\prime n})_{\omega^{\otimes n}},
\end{equation}
where $\frac{1}{n}\log_{2}M$ is understood as the non-asymptotic entanglement
cost of the protocol, and $E_{F,0}^{\varepsilon/2}(A^{\prime n};B^{\prime
n})_{\omega^{\otimes n}}$ is the $\varepsilon/2$-smooth entanglement of
formation (EoF) \cite{BD11} recalled in Definition~\ref{def:smooth-EoF}\ below.
\end{proposition}

\begin{definition}
[Smooth EoF \cite{BD11}]\label{def:smooth-EoF}Let $\delta\in\left(
0,1\right)  $ and $\tau_{CD}$ be a bipartite state. Let $\mathcal{E}=\left\{
p_{X}(x),\phi_{CD}^{x}\right\}  $ denote a pure-state ensemble decomposition
of $\tau_{CD}$, meaning that $\tau_{CD}=\sum_{x}p_{X}(x)\phi_{CD}^{x}$, where
$\phi_{CD}^{x}$ is a pure state and $p_{X}$ is a probability distribution.
Define the conditional entropy of order zero $H_{0}(K|L)_{\omega}$ of a
bipartite state $\omega_{KL}$ as%
\begin{equation}
H_{0}(K|L)_{\omega}\equiv\max_{\sigma_{L}}  \log_{2}\operatorname{Tr}%
\{\Pi_{KL}^{\omega}(I_{K}\otimes\sigma_{L})\}  ,
\end{equation}
where $\Pi_{KL}^{\omega}$ denotes the projection onto the support of
$\omega_{KL}$ and $\sigma_{L}$ is a density operator. Then the $\delta$-smooth
entanglement of formation of $\tau_{CD}$ is given by%
\begin{equation}
E_{F,0}^{\delta}(C;D)_{\tau}\equiv\min_{\mathcal{E},\widetilde{\tau}_{XC}\in
B_{cq}^{\delta}(\tau_{XC})}H_{0}(C|X)_{\widetilde{\tau}},
\end{equation}
where the minimization is with respect to all pure-state ensemble decompositions
$\mathcal{E}$ of $\tau_{CD}$, $\tau_{XCD}=\sum_{x}p_{X}(x)|x\rangle\langle
x|_{X}\otimes\phi_{CD}^{x}$ is a labeled pure-state extension of $\tau_{CD}$,
and the $\delta$-ball $B_{cq}^{\delta}(\tau_{XC})$ of cq states for a cq state
$\tau_{XC}$ is defined as%
\begin{multline}
B_{cq}^{\delta}(\tau_{XC})\equiv\Big\{\omega_{XC}:\omega_{XC}\geq
0,\ \omega_{XC}=\sum_{x}|x\rangle\langle x|\otimes\omega_{C}^{x},\\
\left\Vert \omega_{XC}-\tau_{XC}\right\Vert _{1}\leq\delta\Big\}.
\end{multline}
The $\delta$-smooth entanglement of formation has the property that, for a
tensor-power state $\tau_{CD}^{\otimes n}$, the following limit holds
\cite[Theorem~2]{BD11}%
\begin{align}
\lim_{\delta\rightarrow0}\lim_{n\rightarrow\infty}\frac{1}{n}E_{F,0}^{\delta
}(C^{n};D^{n})_{\tau^{\otimes n}} &  =\lim_{n\rightarrow\infty}\frac{1}%
{n}E_{F}(C;D)_{\tau},\label{eq:one-shot-cost-tens-pow}\\
&  = E_{C}(\tau_{CD}).
\end{align}
where the latter quantity denotes the entanglement cost of the state
$\tau_{CD}$ \cite{HHT01}.
\end{definition}

\begin{proof}
[Proof of Proposition~\ref{prop:ach-TP-sim}]The approach for an
$(n,M,\varepsilon)$\ sequential channel simulation consists of the following steps:

First, employ the one-shot entanglement cost protocol from \cite[Theorem~1]%
{BD11}, which consumes a maximally entangled state $\Phi_{\overline{A}%
_{0}\overline{B}_{0}}$ of Schmidt rank $M$ along with an LOCC\ channel
$\mathcal{P}_{\overline{A}_{0}\overline{B}_{0}\rightarrow A^{\prime
n}B^{\prime n}}$ to generate $n$ approximate copies of the resource state
$\omega_{A^{\prime}B^{\prime}}$. In particular, using the maximally entangled
state $\Phi_{\overline{A}_{0}\overline{B}_{0}}$ with%
\begin{equation}
\log_{2}M=E_{F,0}^{\varepsilon/2}(A^{\prime n};B^{\prime n})_{\omega^{\otimes
n}},
\end{equation}
one can achieve the following approximation \cite[Theorem~1]{BD11}:%
\begin{equation}
\frac{1}{2}\left\Vert \omega_{A^{\prime}B^{\prime}}^{\otimes n}-\widetilde
{\omega}_{A^{\prime n}B^{\prime n}}\right\Vert _{1}\leq\sqrt{\varepsilon},
\end{equation}
where%
\begin{equation}
\widetilde{\omega}_{A^{\prime n}B^{\prime n}}\equiv\mathcal{P}_{\overline
{A}_{0}\overline{B}_{0}\rightarrow A^{\prime n}B^{\prime n}}(\Phi
_{\overline{A}_{0}\overline{B}_{0}}).
\end{equation}

Next, at the first instance in which the channel should be simulated, Alice
and Bob apply the LOCC channel $\mathcal{L}_{AA^{\prime}B^{\prime}\rightarrow
B}$ from \eqref{eq:tp-sim} to the $A_{1}^{\prime}$ and $B_{1}^{\prime}$
systems of $\widetilde{\omega}_{A^{\prime n}B^{\prime n}}$. For the second
instance, they apply the LOCC channel $\mathcal{L}_{AA^{\prime}B^{\prime
}\rightarrow B}$ from \eqref{eq:tp-sim} to the $A_{2}^{\prime}$ and
$B_{2}^{\prime}$ systems of $\widetilde{\omega}_{A^{\prime n}B^{\prime n}}$.
This continues for the next $n-2$ rounds of the sequential channel simulation.

By the data processing inequality for trace distance, it is guaranteed that
the following bound holds on the performance of this protocol for sequential
channel simulation:%
\begin{equation}
\frac{1}{2}\left\Vert (\mathcal{N})^{n}-(\mathcal{L})^{n}\right\Vert
_{\Diamond,n}\leq\frac{1}{2}\left\Vert \omega_{A^{\prime}B^{\prime}}^{\otimes
n}-\widetilde{\omega}_{A^{\prime n}B^{\prime n}}\right\Vert _{1}\leq
\sqrt{\varepsilon}. \label{eq:limitation-on-dist-seq-sim}%
\end{equation}
This follows because the distinguishability of the simulation from the actual
channel uses is limited by the distinguishability of the states $\omega
_{A^{\prime}B^{\prime}}^{\otimes n}$ and $\widetilde{\omega}_{A^{\prime
n}B^{\prime n}}$, due to the assumed structure of the channel in
\eqref{eq:tp-sim}, as well as the structure of the sequential channel simulation.
\end{proof}

By applying definitions, the bound in Proposition~\ref{prop:ach-TP-sim},
taking the limits $n\rightarrow\infty$ and then $\varepsilon\rightarrow0$
(with $M=2^{n[R+\delta]}$ for a fixed rate $R$ and arbitrary $\delta>0$), and
applying \eqref{eq:one-shot-cost-tens-pow}, we conclude the following statement:

\begin{corollary}
Let $\mathcal{N}_{A\rightarrow B}$ be a quantum channel that is
teleportation-simulable with associated resource state $\omega_{A^{\prime
}B^{\prime}}$, as defined in \eqref{eq:tp-sim}. Then the entanglement cost of
the channel $\mathcal{N}$ is never larger than the entanglement cost of the
resource state $\omega_{A^{\prime}B^{\prime}}$:%
\begin{equation}
E_{C}(\mathcal{N})\leq E_{C}(\omega_{A^{\prime}B^{\prime}}).
\label{eq:channel-cost-less-than-resource-cost}%
\end{equation}

\end{corollary}

The above corollary captures the intuitive idea that if a single instance of
the channel $\mathcal{N}$ can be simulated via LOCC\ starting from a resource
state $\omega_{A^{\prime}B^{\prime}}$, then the entanglement cost of the
channel should not exceed the entanglement cost of the resource state. The
idea of the above proof is simply to prepare a large number $n$ of copies of
$\omega_{A^{\prime}B^{\prime}}$ approximately and then use these to simulate $n$ uses of the
channel $\mathcal{N}$, such that the simulation could not be distinguished
from $n$ uses of the channel $\mathcal{N}$ in any sequential test.

\subsection{The entanglement cost of resource-seizable,
teleportation-simulable channels}

\label{sec:eq-res-seize}In this section, I define teleportation-simulable
channels that are resource-seizable, meaning that one can seize the channel's
underlying resource state by the following procedure:

\begin{enumerate}
\item prepare a free, separable state,

\item input one of its systems to the channel, and then

\item post-process with a free, LOCC channel.
\end{enumerate}

\noindent This procedure is indeed related to the channel processing described
earlier in \eqref{eq:pre-post}. After that, I prove that the entanglement cost
of a resource-seizable channel is equal to the entanglement cost of its
underlying resource state.

\begin{definition}
[Resource-seizable channel]\label{def:res-seiz-ch}Let $\mathcal{N}%
_{A\rightarrow B}$ be a teleportation-simulable channel with associated
resource state $\omega_{A^{\prime}B^{\prime}}$, as defined in
\eqref{eq:tp-sim}. Suppose that there exists a separable input state
$\rho_{A_{M}AB_{M}}$ to the channel and a postprocessing LOCC channel
$\mathcal{D}_{A_{M}BB_{M}\rightarrow A^{\prime}B^{\prime}}$\ such that the
resource state $\omega_{A^{\prime}B^{\prime}}$ can be seized from the channel
$\mathcal{N}_{A\rightarrow B}$ as follows:%
\begin{equation}
\mathcal{D}_{A_{M}BB_{M}\rightarrow A^{\prime}B^{\prime}}(\mathcal{N}%
_{A\rightarrow B}(\rho_{A_{M}AB_{M}}))=\omega_{A^{\prime}B^{\prime}}.
\label{eq:seize}%
\end{equation}
Then we say that the channel is a resource-seizable, teleportation-simulable channel.
\end{definition}

In Appendix~\ref{app:res-seize-impl-image}, I discuss how resource-seizable channels are related to those
that are \textquotedblleft implementable from their image,\textquotedblright%
\ as defined in \cite[Appendix~A]{CM17}. In
Section~\ref{sec:resource-theory-gen}, I also discuss how to generalize the
notion of a resource-seizable channel to an arbitrary resource theory.

The main result of this section is the following simplifying form for the
entanglement cost of a resource-seizable channel (as defined above),
establishing that its entanglement cost in the asymptotic regime is the same
as the entanglement cost of the underlying resource state. Furthermore, for
these channels, the entanglement cost is not increased by the need to pass a
more stringest test for channel simulation as required in a sequential test.

\begin{theorem}
\label{thm:tp-sim}Let $\mathcal{N}_{A\rightarrow B}$ be a resource-seizable,
teleportation-simulable channel with associated resource state $\omega
_{A^{\prime}B^{\prime}}$, as given in Definition~\ref{def:res-seiz-ch}. Then
the entanglement cost of the channel $\mathcal{N}_{A\rightarrow B}$ is equal
to its parallel entanglement cost, which in turn is equal to the entanglement
cost of the resource state $\omega_{A^{\prime}B^{\prime}}$:%
\begin{equation}
E_{C}(\mathcal{N})=E_{C}^{(p)}(\mathcal{N})=E_{C}(\omega_{A^{\prime}B^{\prime
}}).
\end{equation}

\end{theorem}

\begin{proof}
Consider from \eqref{eq:sequential-greater-than-para}\ that%
\begin{equation}
E_{C}(\mathcal{N})\geq E_{C}^{(p)}(\mathcal{N})=\lim_{n\rightarrow\infty}%
\frac{1}{n}E_{F}(\mathcal{N}^{\otimes n}).
\label{eq:seq-greater-than-parallel-1}%
\end{equation}
Let $\psi_{RA^{n}}\equiv\psi_{RA_{1}\cdots A_{n}}$ be an arbitrary pure input
state to consider at the input of the tensor-power channel $(\mathcal{N}%
_{A\rightarrow B})^{\otimes n}$, leading to the state%
\begin{equation}
\sigma_{RB^{n}}\equiv(\mathcal{N}_{A\rightarrow B})^{\otimes n}(\psi
_{RA_{1}\cdots A_{n}}).
\end{equation}
From the assumption that the channel is teleportation-simulable with
associated resource state $\omega_{A^{\prime}B^{\prime}}$, we have from
\eqref{eq:tp-sim}\ that%
\begin{equation}
\sigma_{RB^{n}}=(\mathcal{L}_{AA^{\prime}B^{\prime}\rightarrow B})^{\otimes
n}(\psi_{RA^{n}}\otimes\omega_{A^{\prime}B^{\prime}}^{\otimes n})
\end{equation}
Then%
\begin{align}
E_{F}(R;B^{n})_{\sigma}  &  \leq E_{F}(RA^{n}A^{\prime n};B^{\prime n}%
)_{\psi\otimes\omega^{\otimes n}}\\
&  =E_{F}(A^{\prime n};B^{\prime n})_{\omega^{\otimes n}},
\end{align}
where the inequality follows from LOCC\ monotonicity of the entanglement of
formation. Since the bound holds for an arbitrary input state, we conclude
that the following inequality holds for all $n\in\mathbb{N}$:%
\begin{equation}
\frac{1}{n}E_{F}(\mathcal{N}^{\otimes n})\leq\frac{1}{n}E_{F}(A^{\prime
n};B^{\prime n})_{\omega^{\otimes n}}.
\end{equation}
Now taking the limit $n\rightarrow\infty$, we conclude that%
\begin{equation}
E_{C}^{(p)}(\mathcal{N})\leq E_{C}(\omega_{A^{\prime}B^{\prime}}).
\label{eq:Tp-sim-ineq}%
\end{equation}

To see the other inequality, let a decomposition of the separable input state
$\rho_{A_{M}AB_{M}}$ be given by%
\begin{equation}
\rho_{A_{M}AB_{M}}=\sum_{x}p_{X}(x)\psi_{A_{M}A}^{x}\otimes\phi_{B_{M}}^{x}.
\end{equation}
Considering that $[\psi_{A_{M}A}^{x}]^{\otimes n}$ is a particular input to
the tensor-power channel $(\mathcal{N}_{A\rightarrow B})^{\otimes n}$, we
conclude that%
\begin{equation}
E_{F}(\mathcal{N}^{\otimes n})\geq E_{F}(A_{M}^{n};B^{n})_{[\mathcal{N}%
(\psi^{x})]^{\otimes n}}.
\end{equation}
Since this holds for all $x$, we have that%
\begin{align}
E_{F}(\mathcal{N}^{\otimes n})  &  \geq\sum_{x}p_{X}(x)E_{F}(A_{M}^{n}%
;B^{n})_{[\mathcal{N}(\psi^{x})]^{\otimes n}}\nonumber\\
&  =\sum_{x}p_{X}(x)E_{F}(A_{M}^{n};B^{n}B_{M}^{n})_{[\mathcal{N}(\psi
^{x})\otimes\phi^{x}]^{\otimes n}}\nonumber\\
&  \geq E_{F}(A_{M}^{n};B^{n}B_{M}^{n})_{[\mathcal{N}(\rho)]^{\otimes n}%
}\nonumber\\
&  \geq E_{F}(A^{\prime n};B^{\prime n})_{\omega^{\otimes n}},
\end{align}
where the equality follows because introducing a product state locally does
not change the entanglement, the second inequality follows from convexity of
entanglement of formation \cite{BDSW96}, and the last inequality follows from
the assumption in \eqref{eq:seize} and the LOCC\ monotonicity of the
entanglement of formation. Since the inequality holds for all $n\in\mathbb{N}%
$, we can divide by $n$ and take the limit $n\rightarrow\infty$ to conclude
that%
\begin{equation}
E_{C}^{(p)}(\mathcal{N})\geq E_{C}(\omega_{A^{\prime}B^{\prime}}),
\end{equation}
and in turn, from \eqref{eq:Tp-sim-ineq}, that%
\begin{equation}
E_{C}^{(p)}(\mathcal{N})=E_{C}(\omega_{A^{\prime}B^{\prime}}).
\end{equation}
Combining this equality with the inequalities in
\eqref{eq:channel-cost-less-than-resource-cost} and
\eqref{eq:seq-greater-than-parallel-1}\ leads to the statement of the theorem.
\end{proof}

\section{Examples}

\label{sec:examples}

The equality in Theorem~\ref{thm:tp-sim} provides a formal expression for the
entanglement cost of any resource-seizable, teleportation-simulable channel,
given in terms of the entanglement cost of the underlying resource state
$\omega_{A^{\prime}B^{\prime}}$. Due to the fact that the entanglement cost of
a state is generally not equal to its entanglement of formation \cite{H09}, it
could still be a significant challenge to compute the entanglement cost of
these special channels. However, for some special states, the equality
$E_{C}(\omega_{A^{\prime}B^{\prime}})=E_{F}(A^{\prime};B^{\prime})_{\omega}$
does hold, and I discuss several of these examples and related channels here.

Let us begin by recalling the notion of a covariant channel $\mathcal{N}%
_{A\rightarrow B}$\ \cite{Hol02}. For a group $G$ with unitary channel
representations $\{\mathcal{U}_{A}^{g}\}_{g}$ and $\{\mathcal{V}_{B}^{g}%
\}_{g}$ acting on the input system $A$ and output system $B$ of the channel
$\mathcal{N}_{A\rightarrow B}$, the channel $\mathcal{N}_{A\rightarrow B}$ is
covariant with respect to the group$~G$ if the following equality holds%
\begin{equation}
\mathcal{N}_{A\rightarrow B}\circ\mathcal{U}_{A}^{g}=\mathcal{V}_{B}^{g}%
\circ\mathcal{N}_{A\rightarrow B}.
\label{eq:cov-to-help}
\end{equation}
If the averaging channel is such that $\frac{1}{\left\vert G\right\vert }%
\sum_{g}\mathcal{U}_{A}^{g}(X)=\operatorname{Tr}[X]I/\left\vert A\right\vert $
(implementing a unitary one-design), then we simply say that the channel
$\mathcal{N}_{A\rightarrow B}$ is covariant.

Then from \cite[Section~7]{CDP09} (see also \cite[Appendix~A]{WTB16}), we
conclude that any covariant channel is teleportation-simulable with associated
resource state given by the Choi state of the channel, i.e., $\omega
_{A^{\prime}B^{\prime}}=\mathcal{N}_{A\rightarrow B}(\Phi_{A^{\prime}A})$. As
such, covariant channels are resource-seizable, so that the equality in
Theorem~\ref{thm:tp-sim}\ applies to all covariant channels.
\textit{\textbf{Thus, the entanglement cost of a covariant channel is equal to
the entanglement cost of its Choi state.}} In spite of this reduction, it
could still be a great challenge to compute formulas for the entanglement cost
of these channels, due to the fact that the entanglement of formation is not
necessarily equal to the entanglement cost for the Choi states of these
channels. For example, the entanglement cost of an isotropic state \cite{W89,PhysRevA.59.4206},
which is the Choi state of a depolarizing channel, is not known. In the next few subsections, I detail some example channels for which it is possible to characterize their entanglement cost.

\subsection{Erasure channels}

A simple example of a channel that is covariant is the quantum erasure
channel, defined as \cite{GBP97}%
\begin{equation}
\mathcal{E}^{q}(\rho)\equiv(1-q)\rho+q|e\rangle\langle e|,
\end{equation}
where $\rho$ is a $d$-dimensional input state, $q\in\left[  0,1\right]  $ is
the erasure probability, and $|e\rangle\langle e|$ is a pure erasure state
orthogonal to any input state, so that the output state has $d+1$ dimensions.
By the remark above, we conclude that $E_{C}(\mathcal{E}^{q})=E_{C}%
(\mathcal{E}_{A\rightarrow B}^{q}(\Phi_{RA}))$, and so determining the
entanglement cost boils down to determining the entanglement cost of the Choi
state%
\begin{equation}
\mathcal{E}_{A\rightarrow B}^{q}(\Phi_{RA})=\left(  1-q\right)  \Phi
_{RA}+\frac{I_{R}}{d}\otimes|e\rangle\langle e|.
\end{equation}
An obvious pure-state decomposition for $\mathcal{E}_{A\rightarrow B}^{q}%
(\Phi_{RA})$ (see \cite[Eqs.~(93)--(95)]{BBCW13}) leads to%
\begin{align}
E_{C}(\mathcal{E}_{A\rightarrow B}^{q}(\Phi_{RA}))  &  \leq E_{F}%
(\mathcal{E}_{A\rightarrow B}^{q}(\Phi_{RA}))\\
&  \leq\left(  1-q\right)  \log_{2}d.
\end{align}
As it turns out, these inequalities are tight, due to an operational argument. In
particular, the distillable entanglement of $\mathcal{E}_{A\rightarrow B}%
^{q}(\Phi_{RA})$ is exactly equal to $\left(  1-q\right)  \log_{2}d$
\cite{PhysRevLett.78.3217}, and due to the operational fact that the
distillable entanglement of a state cannot exceed its entanglement cost
\cite{BDSW96}, we conclude that $E_{C}(\mathcal{E}_{A\rightarrow B}^{q}%
(\Phi_{RA}))=\left(  1-q\right)  \log_{2}d$, and in turn that%
\begin{equation}
E_{C}(\mathcal{E}^{q})=E_{C}^{(p)}(\mathcal{E}^{q})=\left(  1-q\right)
\log_{2}d.
\end{equation}
This result generalizes the finding from \cite{BBCW13}, which is that
$E_{C}^{(p)}(\mathcal{E}^{q})=\left(  1-q\right)  \log_{2}d$, and so we
conclude that for erasure channels, the entanglement cost of these channels is
not increased by the need to pass a more stringest test for channel
simulation, as posed by a sequential test.
Note also that the distillable entanglement of the erasure channel is given by $E_D(\mathcal{E}^{q})=\left(  1-q\right)  \log_{2}d$, due to \cite{PhysRevLett.78.3217}.

The fact that the distillable entanglement of an erasure channel is equal to
its entanglement cost, implies that, if we restrict the resource theory of
entanglement for
quantum channels to consist solely of erasure channels, then it is reversible.
By this, we mean that, in the limit of many channel uses, if one begins with
an erasure channel of parameter $q$ and distills ebits from it at a rate
$(1-q)\log_{2} d$, then one can subsequently use these distilled ebits to
simulate the same erasure channel again. As we see below, this reversibility
breaks down when considering other channels.

\subsection{Dephasing channels}

A $d$-dimensional dephasing channel has the following action:%
\begin{equation}
\mathcal{D}^{\mathbf{q}}(\rho)=\sum_{i=0}^{d-1}q_{i}Z^{i}\rho Z^{i\dag},
\end{equation}
where $\mathbf{q}$ is a vector containing the probabilities $q_{i}$ and $Z$
has the following action on the computational basis $Z|x\rangle=e^{2\pi
ix/d}|x\rangle$. This channel is covariant with respect to the
Heisenberg--Weyl group of unitaries, which are well known to be a unitary
one-design. Furthermore, as remarked previously (e.g., in \cite{TWW17}), the
Choi state $\mathcal{D}_{A\rightarrow B}^{\mathbf{q}}(\Phi_{RA})$ of this
channel is a maximally correlated state \cite{Rai99,Rai01}, which has the form%
\begin{equation}
\sum_{i,j}\alpha_{i,j}|i\rangle\langle j|_{R}\otimes|i\rangle\langle j|_{B}.
\end{equation}
As such, Theorem~\ref{thm:tp-sim} applies to these channels, implying that%
\begin{align}
E_{C}(\mathcal{D}^{\mathbf{q}})  &  =E_{C}^{(p)}(\mathcal{D}^{\mathbf{q}%
})=E_{C}(\mathcal{D}_{A\rightarrow B}^{\mathbf{q}}(\Phi_{RA}))\\
&  =E_{F}(\mathcal{D}_{A\rightarrow B}^{\mathbf{q}}(\Phi_{RA})),
\end{align}
with the final equality resulting from the fact that the entanglement cost is
equal to the entanglement of formation for maximally correlated states
\cite{VDC02,HSS03}. In \cite[Section~VI-A]{HSS03}, an optimization procedure
is given for calculating the entanglement of formation of maximally correlated
states, which is simpler than that needed from the definition of entanglement
of formation.

A qubit dephasing channel with a single dephasing parameter $q\in\left[
0,1\right]  $ is defined as%
\begin{equation}
\mathcal{D}^{q}(\rho)=\left(  1-q\right)  \rho+qZ\rho Z.
\end{equation}
For the Choi state of this channel, there is an explicit formula for its
entanglement of formation \cite{W98}, from which we can conclude that%
\begin{equation}
E_{C}(\mathcal{D}^{q})=E_{C}^{(p)}(\mathcal{D}^{q})=h_{2}(1/2+\sqrt{q\left(
1-q\right)  }), \label{eq:qubit-dephasing-cost}%
\end{equation}
where
\begin{equation}
h_{2}(x)\equiv-x\log_{2}x-\left(  1-x\right)  \log_{2}(1-x) \label{eq:bin-ent}%
\end{equation}
is the binary entropy. The equality in \eqref{eq:qubit-dephasing-cost} solves
an open question from \cite{BBCW13}, where it had only been shown that
$E_{C}^{(p)}(\mathcal{D}^{q})\leq h_{2}(1/2+\sqrt{q\left(  1-q\right)  })$.

The results of \cite[Eq.~(57)]{BDSW96} and \cite[Eq.~(8.114)]{H06book} gave a
simple formula for the distillable entanglement of the qubit dephasing
channel:
\begin{equation}
E_{D}(\mathcal{D}^{q}) =1-h_{2}(q).
\end{equation}
Thus, this formula and the formula in \eqref{eq:qubit-dephasing-cost}
demonstrate that the resource theory of entanglement for these channels is
irreversible. That is, if one started from a qubit dephasing channel with
parameter $q\in(0,1)$ and distilled ebits from it at the ideal rate of
$1-h_{2}(q)$, and then subsequently wanted to use these ebits to simulate a
qubit dephasing channel with the same parameter, this is not possible, because
the rate at which ebits are distilled is not sufficient to simulate the
channel again. Figure~\ref{fig:EoF-dephasing} compares the formulas for
entanglement cost and distillable entanglement of the qubit dephasing channel,
demonstrating that there is a noticeable gap between them. At $q=1/2$, the qubit dephasing channel is a completely dephasing, classical channel, so that
$E_{C}(\mathcal{D}^{1/2}) = E_{D}(\mathcal{D}^{1/2})=0$. Thus, a reasonable approximation to the difference is given by a Taylor expansion about $q=1/2$:
\begin{multline}
E_{C}(\mathcal{D}^{q}) -
E_{D}(\mathcal{D}^{q}) = \\
\frac{1}{\ln 2}\left[2\ln\!\left(\frac{1}{|q-\tfrac12|}\right)- 1\right](q-\tfrac12)^2 + O((q-\tfrac12)^4).
\end{multline}

\begin{figure}[ptb]
\begin{center}
\includegraphics[width=3.4in]
{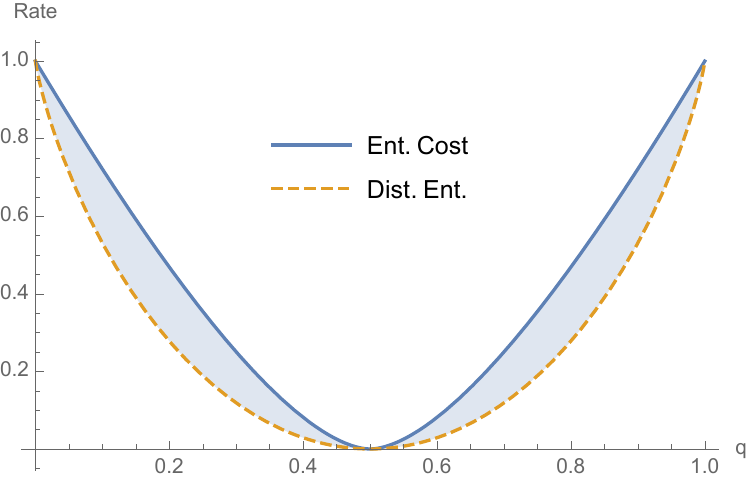}
\end{center}
\caption{Entanglement cost $E_{C}(\mathcal{D}^{q})=h_{2}(1/2+\sqrt{q\left(
1-q\right)  })$ and distillable entanglement $E_{D}(\mathcal{D}^{q})
=1-h_{2}(q)$ of the qubit dephasing channel $\mathcal{D}^{q}$ as a function of
the dephasing parameter $q \in[0,1]$, with the shaded area demonstrating the
gap between them. The units for rate on the vertical axis are ebits per channel use, and $q$ on the horizontal axis is dimensionless.}%
\label{fig:EoF-dephasing}%
\end{figure}

\subsection{Werner--Holevo channels}

A particular kind of Werner--Holevo channel performs the following
transformation on a $d$-dimensional input state $\rho$ \cite{WH02}:%
\begin{equation}
\mathcal{W}^{(d)}(\rho)\equiv\frac{1}{d-1}\left(  \operatorname{Tr}%
\{\rho\}I-T(\rho)\right)  ,
\end{equation}
where $T$ denotes the transpose map $T(\cdot)=\sum_{i,j}|i\rangle\langle
j|(\cdot)|i\rangle\langle j|$. As observed in \cite[Section~II]{WH02} and \cite[Section~VII]{LM15}, this
channel is covariant, and so an immediate consequence of \cite[Section~7]%
{CDP09} is that these channels are teleportation simulable with associated
resource state given by their Choi state. The latter fact was explicitly
observed in \cite[Sections VI and~VII]{LM15}, as well as \cite[Appendix~A]{CM17}. Furthermore, its Choi state is
given by%
\begin{equation}
\mathcal{W}_{A\rightarrow B}^{(d)}(\Phi_{RA})=\alpha_{d}\equiv\frac
{1}{d\left(  d-1\right)  }\left(  I_{RB}-F_{RB}\right)  ,
\label{eq:Choi-WH}
\end{equation}
where $\alpha_{d}$ is the antisymmetric state, i.e., the maximally mixed state
on the antisymmetric subspace of a $d\times d$ quantum system and
$F_{RB}\equiv\sum_{i,j}|i\rangle\langle j|_{R}\otimes|j\rangle\langle i|_{B}$
denotes the unitary swap operator. Theorem~\ref{thm:tp-sim} thus applies to
these channels, and we find that%
\begin{align}
E_{C}(\mathcal{W}^{(d)})  &  =E_{C}^{(p)}(\mathcal{W}^{(d)})\\
&  =E_{C}(\alpha_{d})\\
&  \geq\log_{2}(4/3)\approx0.415,
\label{eq:lower-bnd-CVW-cost-WH}
\end{align}
with the inequality following from \cite[Theorem~2]{CSW12}.\ We also have that%
\begin{equation}
E_{C}(\mathcal{W}^{(d)})=E_{C}^{(p)}(\mathcal{W}^{(d)})=E_{C}(\alpha_{d})\leq
E_{F}(\alpha_{d})=1,
\end{equation}
with the last equality following from the result stated in \cite[Section~IV-C]%
{VW01}. For $d=3$, the entanglement cost $E_{C}(\alpha_{3})$ is known to be
equal to exactly one ebit \cite{F03}:%
\begin{equation}
E_{C}(\mathcal{W}^{(3)})=E_{C}^{(p)}(\mathcal{W}^{(3)})=1.
\label{eq:yura}
\end{equation}

It was observed in \cite[Appendix~A]{CM17} (as well as \cite{W16priv}) that the distillable entanglement of the Werner--Holevo channel $\mathcal{W}^{(d)}$ is equal to the distillable entanglement of its Choi state:
\begin{equation}
E_D(\mathcal{W}^{(d)}) = E_D(\alpha_d).
\end{equation}
Thus, an immediate consequence of \cite[Theorem~1 and Eq.~(5)]{CSW12} is that
\begin{align}
\label{eq:CVW-bnd}
E_D(\mathcal{W}^{(d)})
& \leq  \left. \begin{cases}
                               \log_2\frac{d+2}{d} & \text{ if } d \text{ is even}    \\
                               \frac{1}{2}\log_2\frac{d+3}{d-1} & \text{ if } d \text{ is odd}
                             \end{cases}
                       \right\} \\
                       & = \frac{2}{d\cdot \ln 2} \left(1-\frac{1}{d}\right)
                       + O\!\left(\frac{1}{d^3}\right).
\label{eq:CVW-Taylor-expand}
\end{align}

We can now observe that the resource theory of entanglement is generally not reversible when restricted to Werner--Holevo channels. 
The case $d=2$ is somewhat trivial: in this case, one can verify that the channel $\mathcal{W}^{(2)}$ is a unitary channel, equivalent to acting on the input state with the Pauli $Y$ unitary. Thus, for $d=2$, the channel is a noiseless qubit channel, and we trivially have that 
\begin{equation}
E_D(\mathcal{W}^{(2)}) =  
E_C(\mathcal{W}^{(2)}) = 1,
\label{eq:exact-WH-d=2}
\end{equation}
so that the resource theory of entanglement is clearly reversible in this case.
For $d=3$, the upper bound on distillable entanglement in \eqref{eq:CVW-bnd} evaluates to $\frac{1}{2}\log_2 (3) \approx 0.793$, while the entanglement cost is equal to one, as stated in \eqref{eq:yura}, so that
\begin{equation}
E_D(\mathcal{W}^{(3)}) \leq 
0.793 < 1 = E_C(\mathcal{W}^{(3)}).
\end{equation}
Thus, the resource theory of entanglement is not reversible for $\mathcal{W}^{(3)}$.
For $d\in\{4,5,6\}$, the upper bound in 
\eqref{eq:CVW-bnd} and the lower bound in 
\eqref{eq:lower-bnd-CVW-cost-WH} are not strong enough to make a definitive statement (interestingly, the bounds in 
\eqref{eq:CVW-bnd} and  
\eqref{eq:lower-bnd-CVW-cost-WH} are actually equal for $d=6$). Then for $d\geq 7$, the upper bound in 
\eqref{eq:CVW-bnd} and the lower bound in 
\eqref{eq:lower-bnd-CVW-cost-WH} are strong enough to conclude that
\begin{equation}
E_D(\mathcal{W}^{(d)}) < E_C(\mathcal{W}^{(d)}),
\end{equation}
so that the resource theory is not reversible for $\mathcal{W}^{(d)}$. Figure~\ref{fig:WH-ch} summarizes these observations.

\begin{figure}[ptb]
\begin{center}
\includegraphics[width=3.4in]
{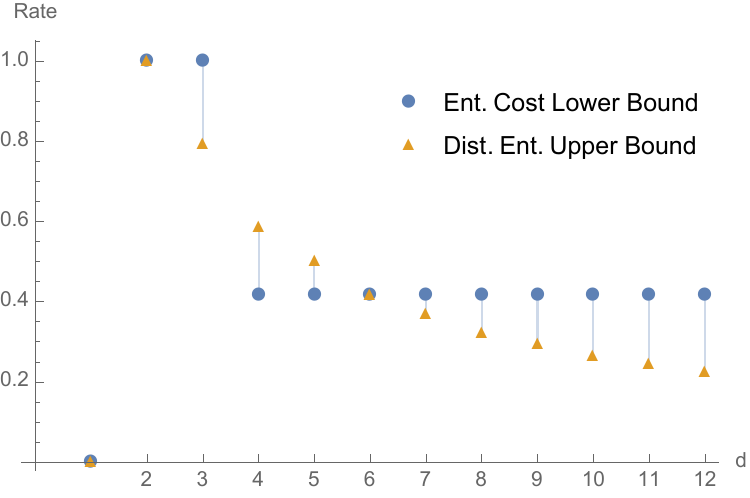}
\end{center}
\caption{Lower bound on the entanglement cost $E_{C}(\mathcal{W}^{(d)})$ from \eqref{eq:lower-bnd-CVW-cost-WH}  and upper bound on distillable entanglement $E_{D}(\mathcal{W}^{(d)})
$ from \eqref{eq:CVW-bnd} for the Werner--Holevo channel
$\mathcal{W}^{(d)}$ as a function of
the parameter $d\geq 4$, with the lines connecting the dots demonstrating the
gap between them. For $d=2$, the points are exact due to \eqref{eq:exact-WH-d=2}, and reversibility holds. For $d=3$,
the entanglement cost
$E_{C}(\mathcal{W}^{(3)})$
is exactly equal to one, as recalled in \eqref{eq:yura}, while
\eqref{eq:CVW-bnd} applies to
$E_{D}(\mathcal{W}^{(3)})$,
and the resource theory is irreversible.
For $d\in \{4,5,6\}$, the bounds are not strong enough to reach a conclusion about reversibility. For $d\geq 7$, the resource theory is irreversible, and the gap
$E_{C}(\mathcal{W}^{(d)})-
E_{D}(\mathcal{W}^{(d)})$
grows at least as large as the difference of \eqref{eq:lower-bnd-CVW-cost-WH} and~\eqref{eq:CVW-Taylor-expand}. The units for rate on the vertical axis are ebits per channel use, and $d$ on the horizontal axis is dimensionless.}
\label{fig:WH-ch}%
\end{figure}

\subsection{Epolarizing channels (complements of depolarizing channels)}

The $d$-dimensional depolarizing channel is a common model of noise in quantum
information, transmitting the input state with probability $1-q\in\left[
0,1\right]  $ and replacing it with the maximally mixed state
$\pi\equiv\frac{I}{d}$
with probability~$q$:%
\begin{equation}
\Delta^{q}(\rho)=\left(  1-q\right)  \rho+q\pi.
\end{equation}
According to
Stinespring's theorem \cite{S55}, every quantum channel $\mathcal{N}_{A\rightarrow B}$
can be realized by the action of some isometric channel $\mathcal{U}%
_{A\rightarrow BE}$ followed by a partial trace:%
\begin{equation}
\mathcal{N}_{A\rightarrow B}(\rho_{A})=\operatorname{Tr}_{E}\{\mathcal{U}%
_{A\rightarrow BE}(\rho_{A})\}.
\label{eq:iso-extend}
\end{equation}
Due to the partial trace and its invariance with respect to isometric channels
acting exclusively on the $E$ system, the extending channel $\mathcal{U}%
_{A\rightarrow BE}$ is not unique in general, but it is unique up to this
freedom. Then given an isometric channel $\mathcal{U}_{A\rightarrow BE}$
extending $\mathcal{N}_{A\rightarrow B}$ as in \eqref{eq:iso-extend}, the complementary channel
$\mathcal{N}_{A\rightarrow E}^{c}$ is defined by a partial trace over the
system $B$ and is interpreted physically as the channel from the input to the
environment:%
\begin{equation}
\mathcal{N}_{A\rightarrow E}^{c}(\rho_{A})=\operatorname{Tr}_{B}%
\{\mathcal{U}_{A\rightarrow BE}(\rho_{A})\}.
\end{equation}

Due to the fact that properties of the original channel are related to
properties of its complementary channel
\cite{KMNR07,ieee2005dev}, there has been significant interest in
understanding complementary channels. In this spirit, and due to the prominent
role of the depolarizing channel, researchers have studied its complementary
channels \cite{DFH06,Leung2017complementary}. In
\cite[Eq.~(3.6)]{DFH06}, the following form was given for a complementary channel of
$\Delta^{q}$:%
\begin{equation}
\rho\rightarrow S_{AF}^{q}\left(  \rho_{A}\otimes I_{F}\right)  S_{AF}^{q\dag
},
\end{equation}
where $I_{F}$ is a $d$-dimensional identity operator and%
\begin{multline}
S_{AF}^{q}\equiv \sqrt{\frac{q}{d}}I_{AF} \ + \\
\sqrt{d}\left(  -\frac{\sqrt{q}}{d}%
+\sqrt{1-q\left(  \frac{d^{2}-1}{d^{2}}\right)  }\right)  \Phi_{AF}.
\end{multline}
A channel complementary to $\Delta^{q}$ has been called an \textquotedblleft
epolarizing channel\textquotedblright\ in \cite{Leung2017complementary}.

An alternative complementary
channel, related to the above one by an isometry acting on the output systems~$AF$, but perhaps more intuitive,
is realized in the following way \cite[Eq.~(28)]{Leung2017complementary}. Consider the isometry $U_{A\rightarrow
SG_{1}G_{2}A}$ defined as%
\begin{multline}
U_{A\rightarrow SG_{1}G_{2}A}|\psi\rangle_{A}\equiv
\\
\text{C-SWAP}_{SG_{1}%
A}\left(  |\phi^{q}\rangle_{S}\otimes|\Phi\rangle_{G_{1}G_{2}}\otimes
|\psi\rangle_{A}\right)  ,
\label{eq:iso-extend-epolar}
\end{multline}
where the control qubit $|\phi^{q}\rangle_{S}\equiv\sqrt{1-q}|0\rangle
_{S}+\sqrt{q}|1\rangle_{S}$,
$|\Phi\rangle_{G_{1}G_{2}}$ is a maximally entangled state of Schmidt rank~$d$,
and the controlled-SWAP\ unitary is given by%
\begin{equation}
\text{C-SWAP}_{SG_{1}A}\equiv|0\rangle\langle0|_{S}\otimes I_{G_{1}%
A}+|1\rangle\langle1|_{S}\otimes\text{SWAP}_{G_{1}A},
\end{equation}
with SWAP$_{G_{1}A}$ denoting a unitary swap operation. By tracing over the
systems $SG_{1}G_{2}$, we recover the original depolarizing channel%
\begin{equation}
\Delta^{q}(\rho_{A})=\operatorname{Tr}_{SG_{1}G_{2}}\{U\rho_{A}U^{\dag}\}.
\end{equation}
Thus, by definition, a channel complementary to $\Delta^{q}$ is realized by%
\begin{equation}
\Lambda_{A\rightarrow SG_{1}G_{2}}^{q}(\rho_{A})\equiv\operatorname{Tr}%
_{A}\{U\rho_{A}U^{\dag}\},
\end{equation}
and in what follows, let us refer to $\Lambda_{A\rightarrow SG_{1}G_{2}}^{q}$
as \textit{the} epolarizing channel.

The isometry $U_{A\rightarrow SG_{1}G_{2}A}$ in \eqref{eq:iso-extend-epolar} is unitarily
covariant, in the sense that for an arbitrary unitary $V_{A}$ acting on the
input, we have that%
\begin{equation}
U_{A\rightarrow SG_{1}G_{2}A}V_{A}=\left(  V_{G_{1}}\otimes\overline{V}%
_{G_{2}}\otimes V_{A}\right)  U_{A\rightarrow SG_{1}G_{2}A}%
,\label{eq:cov-epol-isom}%
\end{equation}
where $\overline{V}$ denotes the complex conjugate of $V$. The identity in
\eqref{eq:cov-epol-isom}
follows because%
\begin{align}
& U_{A\rightarrow SG_{1}G_{2}A}V_{A}|\psi\rangle_{A} \notag\\
& =\text{C-SWAP}_{SG_{1}A}\left(  |\phi^{q}\rangle_{S}|\Phi\rangle_{G_{1}%
G_{2}}V_{A}|\psi\rangle_{A}\right) \notag \\
& =\text{C-SWAP}_{SG_{1}A}\left(  |\phi^{q}\rangle_{S}\left(  V_{G_{1}%
}\overline{V}_{G_{2}}\right)  |\Phi\rangle_{G_{1}G_{2}}V_{A}|\psi\rangle
_{A}\right)  \notag\\
& =\left(  V_{G_{1}}\otimes\overline{V}_{G_{2}}\otimes V_{A}\right)
\text{C-SWAP}_{SG_{1}A}\left(  |\phi^{q}\rangle_{S}|\Phi\rangle_{G_{1}G_{2}%
}|\psi\rangle_{A}\right) \notag \\
& =\left(  V_{G_{1}}\otimes\overline{V}_{G_{2}}\otimes V_{A}\right)
U_{A\rightarrow SG_{1}G_{2}A}|\psi\rangle_{A}.
\end{align}
The above analysis omits some tensor-product symbols for brevity. The third
equality uses the well known fact that $|\Phi\rangle_{G_{1}G_{2}}=\left(
V_{G_{1}}\otimes\overline{V}_{G_{2}}\right)  |\Phi\rangle_{G_{1}G_{2}}$. In
the fourth equality, we have exploited the facts that $\overline{V}_{G_{2}}$
commutes with C-SWAP$_{SG_{1}A}$ and that%
\begin{equation}
\text{SWAP}_{G_{1}A}\left(  V_{G_{1}}\otimes V_{A}\right)  =\left(  V_{G_{1}%
}\otimes V_{A}\right)  \text{SWAP}_{G_{1}A}.
\end{equation}
The covariance in \eqref{eq:cov-epol-isom}\ then implies that the epolarizing
channel is covariant in the following sense:%
\begin{multline}
(\Lambda_{A\rightarrow SG_{1}G_{2}}^{q}\circ\mathcal{V}_{A})(\rho
_{A})\\=(\left(  \mathcal{V}_{G_{1}}\otimes\overline{\mathcal{V}}_{G_{2}%
}\right)  \circ\Lambda_{A\rightarrow SG_{1}G_{2}}^{q})(\rho_{A}),
\end{multline}
where $\mathcal{V}$ denotes the unitary channel realized by the the unitary operator $V$.

As such, by the discussion after \eqref{eq:cov-to-help}, the epolarizing channel is a resource-seizable,
teleportation-simulable channel with associated resource state given by
$\Lambda_{A\rightarrow SG_{1}G_{2}}^{q}(\Phi_{RA})$. Thus, Theorem~\ref{thm:tp-sim} applies
to these channels, implying that the first two of the following equalities hold%
\begin{align}
E_{C}(\Lambda^{q})  & =E_{C}^{(p)}(\Lambda^{q})=E_{C}(\Lambda^{q}(\Phi
_{RA}))\\
& =E_{F}(\Lambda^{q}(\Phi_{RA}))\\
& =-\left(  1-q+\frac{q}{d}\right)  \log_{2}\!\left(  1-q+\frac{q}{d}\right)
\nonumber\\
& \qquad\quad-\left(  d-1\right)  \frac{q}{d}\log_{2}\!\left(  \frac{q}{d}\right)
.\label{eq:explicit-exp-min-out-depo}%
\end{align}
Let us now justify the final two equalities, which give a simple formula for
the entanglement cost of epolarizing channels. First, consider that the Choi state
$\Lambda_{A'\rightarrow SG_{1}G_{2}}^{q}(\Phi_{A'A})$\ of the epolarizing
channel is equal to the state resulting from sending in the maximally mixed
state to the isometric channel $\mathcal{U}_{A\rightarrow SG_{1}G_{2}A}$, defined from~\eqref{eq:iso-extend-epolar}:%
\begin{equation}
\Lambda_{A'\rightarrow SG_{1}G_{2}}^{q}(\Phi_{A'A})=\mathcal{U}_{A\rightarrow
SG_{1}G_{2}A}(\pi_{A}),
\label{eq:max-mixed-iso-epolar}
\end{equation}
where system $A'$ is isomorphic to $A$.
This equality is shown in Appendix~\ref{sec:relation-Choi-comp-max-mixed-iso}. As such, then \cite[Theorem~3]{Matsumoto2004} applies,
as discussed in Example~6 therein, and as a consequence, we can conclude the
second and third equalities in the following, with the bipartite cut of systems taken as $S G_1 G_2 | A$:%
\begin{align}
E_{C}(\Lambda_{A'\rightarrow SG_{1}%
G_{2}}^{q}(\Phi_{A'A}))  & =E_{C}(\mathcal{U}_{A\rightarrow SG_{1}%
G_{2}A}(\pi_{A}))\\
& =E_{F}(\mathcal{U}_{A\rightarrow SG_{1}G_{2}A}(\pi_{A}))\\
& =H_{\min}(\Delta^{q}).
\end{align}
The last line features the minimum output entropy of the depolarizing channel,
which was identified in \cite{K03} and shown to be equal to \eqref{eq:explicit-exp-min-out-depo}.

As discussed in previous examples, it is worthwhile to consider the
reversibility of the resource theory of entanglement for epolarizing channels.
In this spirit, by invoking the covariance of $\Lambda^{q}$, the discussion after \eqref{eq:cov-to-help}, \cite[Eq.~(55)]{BDSW96}, and
\cite[Theorem~4.13]{Rai01}, we find the following bound on the distillable entanglement of the
epolarizing channel~$\Lambda^{q}$:%
\begin{equation}
E_{D}(\Lambda^{q})\leq R(A;SG_{1}G_{2})_{\Lambda^{q}(\Phi)},
\label{eq:rains-bnd}
\end{equation}
where $R(A;SG_{1}G_{2})_{\Lambda^{q}(\Phi)}$ denotes the Rains relative
entropy of the state $\Lambda_{A'\rightarrow SG_{1}G_{2}}^{q}(\Phi_{A'A})$.
Recall that the Rains relative entropy for an arbitrary state $\rho_{AB}$ is
defined as \cite{Rai01}
\begin{equation}
R(A;B)_{\rho}\equiv\min_{\tau_{AB}\in\text{PPT}^{\prime}(A;B)}D(\rho_{AB}%
\Vert\tau_{AB}),
\end{equation}
where the quantum relative entropy is defined as \cite{U62} 
\begin{equation}
D(\rho\Vert\tau)\equiv \operatorname{Tr}\{\rho[\log_2 \rho - \log_2 \tau]\}
\end{equation}
and the Rains set PPT$^{\prime}(A;B)$ \cite{AdMVW02} is given by
\begin{equation}
\text{PPT}^{\prime}(A;B) \equiv \left\{  \tau_{AB}%
:\tau_{AB}\geq0\wedge\left\Vert T_{B}(\tau_{AB})\right\Vert _{1}\leq1\right\}
,
\end{equation}
with $T_{B}$ denoting the partial transpose. Appendix~\ref{sec:matlab-listing} details a Matlab
program taking advantage of recent advances in \cite{Fawzi2018,FF18}, in order  to compute the Rains relative
entropy of any bipartite state.

\begin{figure}[ptb]
\begin{center}
\includegraphics[width=3.6in]
{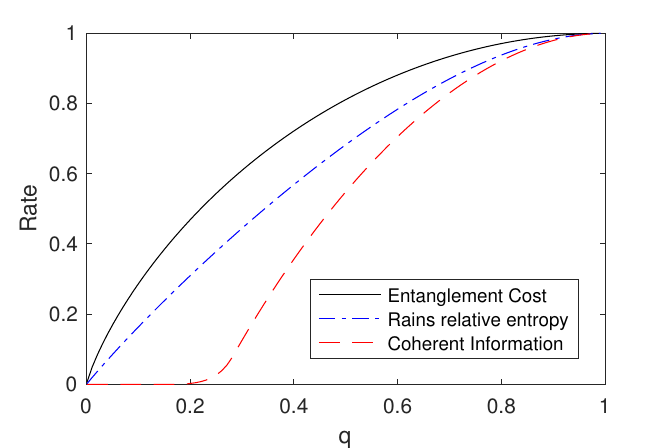}
\end{center}
\caption{The figure depicts the entanglement cost, the Rains bound, and the coherent information of the epolarizing channel $\Lambda^q$, for $d=2$ and $q\in[0,1]$. The gap between the entanglement cost and the Rains bound 
for all $q\in (0,1)$
demonstrates that the resource theory of entanglement is irreversible for epolarizing channels. The units for rate on the vertical axis are ebits per channel use, and $q$ on the horizontal axis is dimensionless.}
\label{fig:epolarizing}%
\end{figure}

Figure~\ref{fig:epolarizing} plots the entanglement cost of the epolarizing channel for $d=2$
(qubit input), and it also plots the Rains bound on distillable entanglement
 in \eqref{eq:rains-bnd}. There is a gap for every value of $q\in(0,1)$, demonstrating that the
resource theory of entanglement is irreversible for epolarizing channels. The
figure also plots the coherent information of the state $\Lambda_{A'\rightarrow
SG_{1}G_{2}}^{q}(\Psi_{A'A}^{s})$, optimized with respect to $|\Psi^{s}%
\rangle_{A'A}\equiv\sqrt{s}|00\rangle_{A'A}+\sqrt{1-s}|11\rangle_{A'A}$ for $s\in[0,1]$, which is
known to be a lower bound on the distillable entanglement of $\Lambda^{q}$ \cite{DW05}.
Note that the coherent information plot is not in contradiction with the
recent result of \cite{Leung2017complementary}, which states that the coherent information is strictly
greater than zero for all $q\in(0,1]$. It is simply that the coherent
information is so small for $q\lesssim0.18$, that it is difficult to witness
its strict positivity numerically.
Matlab files to generate Figure~\ref{fig:epolarizing} are available with the arXiv posting of this paper.

\section{Bosonic Gaussian channels}

\label{sec:BGCs}

In this section, I extend the main ideas of the paper in order to characterize
the entanglement cost of all single-mode bosonic Gaussian channels \cite{S17}.
From a practical perspective, we should be most interested in the single-mode
thermal, amplifier, and additive-noise channels, as these are of the greatest
interest in applications, as stressed in \cite[Section~12.6.3]{H12} and
\cite[Section~3.5]{HG12}. However, it also turns out that these are the only
non-trivial cases to consider among all single-mode bosonic Gaussian channels,
as discussed below.

\subsection{On the definition of entanglement cost for infinite-dimensional
channels}

Before beginning, let us note that there are some subtleties involved when
dealing with quantum information theory in infinite-dimensional Hilbert spaces
\cite{H12}. For example, as advised in \cite{SH08}, the direct use of the
diamond norm in infinite-dimensional Hilbert spaces could be too strong for
applications, and this observation has motivated some recent work
\cite{Sh17,Win17}\ on modifications of the diamond norm that take into account
physical constraints such as energy limitations. On the other hand, the recent
findings in \cite{W18}\ suggest that the direct use of the diamond norm is
reasonable when considering single-mode thermal, amplifier, and additive-noise
channels, as well as some multi-mode bosonic Gaussian channels. As it turns
out, we can indeed directly employ the diamond norm when analyzing the
entanglement cost of these channels. In fact, one of the main contributions of
\cite{W18} was to consider uniform convergence issues in the teleportation
simulation of bosonic Gaussian channels, and due to the fact that the
operational framework of entanglement cost is directly related to the
approximate teleportation simulation of a channel, one should expect that the
findings of \cite{W18} would be related to the issues involved in the
entanglement cost of bosonic Gaussian channels.

With this in mind, let us define the entanglement cost for an
infinite-dimensional channel almost exactly as it has been defined in
Section~\ref{sec:ent-cost-new-def}, with the exception that we allow for LOCC
channels that have a continuous classical index (e.g., as considered in
\cite[Section~4]{S10}), thus going beyond the LOCC\ channels considered in
\eqref{eq:LOCC-channel-decomp}. Specifically, let us define an
$(n,M,\varepsilon)$ sequential channel simulation code as it has been defined
in Section~\ref{sec:ent-cost-new-def}, noting that the $\varepsilon$-error
criterion is given by \eqref{eq:good-sim}, representing the direct
generalization of the strategy norm of \cite{CDP08a,GDP09,G12}\ to
infinite-dimensional systems. Achievable rates and the entanglement cost are
then defined in the same way.

\subsection{Preliminary observations about the entanglement cost of
single-mode bosonic Gaussian channels}

The starting point for our analysis of single-mode bosonic Gaussian channels
is the Holevo classification from \cite{Holevo2007}, in which canonical forms
for all single-mode bosonic Gaussian channels have been given, classifying
them up to local Gaussian unitaries acting on the input and output of the
channel. It then suffices for us to focus our attention on the canonical
forms, as it is self-evident from definitions that local unitaries do not
alter the entanglement cost of a quantum channel. The thermal and amplifier
channels form the class C discussed in \cite{Holevo2007}, and the
additive-noise channels form the class B$_{2}$ discussed in the same work. The
classes that remain are labeled A, B$_{1}$, and D in \cite{Holevo2007}. The
channels in A and D\ are entanglement-breaking \cite{Holevo2008}, and as a
consequence of the \textquotedblleft if-part\textquotedblright\ of
Remark~\ref{rem:ent-break}, they have zero entanglement cost. Channels in the
class B$_{1}$ are perhaps not interesting for practical applications, and as
it turns out, they have infinite quantum capacity \cite{Holevo2007}. Thus,
their entanglement cost is also infinite, because a channel's quantum capacity
is a lower bound on its distillable entanglement, which is in turn a lower
bound on its entanglement cost---these relationships are a direct consequence
of the definitions of the underlying quantities. For the same reason, the
entanglement cost of the bosonic identity channel is also infinite.

\subsection{Thermal, amplifier, and additive-noise channels}

In light of the previous discussion, for the remainder of the paper, let us focus
our attention on the thermal, amplifier, and additive-noise channels. Each of
these are defined respectively by the following Heisenberg input-output
relations:%
\begin{align}
\hat{b}  &  =\sqrt{\eta}\hat{a}+\sqrt{1-\eta}\hat{e}%
,\label{eq:thermal-channel}\\
\hat{b}  &  =\sqrt{G}\hat{a}+\sqrt{G-1}\hat{e}^{\dag}%
,\label{eq:amplifier-channel}\\
\hat{b}  &  =\hat{a}+\left(  x+ip\right)  /\sqrt{2},
\label{eq:additive-noise-channel}%
\end{align}
where $\hat{a}$, $\hat{b}$, and $\hat{e}$ are the field-mode annihilation
operators for the sender's input, the receiver's output, and the environment's
input of these channels, respectively.

The channel in \eqref{eq:thermal-channel} is a thermalizing channel, in which
the environmental mode is prepared in a thermal state $\theta(N_{B})$\ of mean
photon number $N_{B}\geq0$, defined as%
\begin{equation}
\theta(N_{B})\equiv\frac{1}{N_{B}+1}\sum_{n=0}^{\infty}\left(  \frac{N_{B}%
}{N_{B}+1}\right)  ^{n}|n\rangle\langle n|,
\end{equation}
where $\left\{  |n\rangle\right\}  _{n=0}^{\infty}$ is the orthonormal,
photonic number-state basis. When $N_{B}=0$, $\theta(N_{B})$ reduces to the
vacuum state, in which case the resulting channel in
\eqref{eq:thermal-channel} is called the pure-loss channel---it is said to be
quantum-limited in this case because the environment is injecting the minimum
amount of noise allowed by quantum mechanics. The parameter $\eta\in(0,1)$ is
the transmissivity of the channel, representing the average fraction of
photons making it from the input to the output of the channel. Let
$\mathcal{L}_{\eta,N_{B}}$ denote this channel, and we make the further
abbreviation $\mathcal{L}_{\eta}\equiv\mathcal{L}_{\eta,N_{B}=0}$ when it is
the pure-loss channel. The channel in \eqref{eq:thermal-channel} is
entanglement-breaking when $\left(  1-\eta\right)  N_{B}\geq\eta$
\cite{Holevo2008}, and by Remark~\ref{rem:ent-break}, the entanglement cost is
equal to zero for these values.

The channel in \eqref{eq:amplifier-channel} is an amplifier channel, and the
parameter $G>1$ is its gain. For this channel, the environment is prepared in
the thermal state $\theta(N_{B})$. If $N_{B}=0$, the amplifier channel is
called the pure-amplifier channel---it is said to be quantum-limited for a
similar reason as stated above. Let $\mathcal{A}_{G,N_{B}}$ denote this
channel, and we make the further abbreviation $\mathcal{A}_{G}\equiv
\mathcal{A}_{G,N_{B}=0}$ when it is the quantum-limited amplifier channel. The
channel in \eqref{eq:amplifier-channel} is entanglement-breaking when $\left(
G-1\right)  N_{B}\geq1$ \cite{Holevo2008}, and by Remark~\ref{rem:ent-break},
the entanglement cost is equal to zero for these values.

Finally, the channel in \eqref{eq:additive-noise-channel} is an additive-noise
channel, representing a quantum generalization of the classical additive white
Gaussian noise channel. In \eqref{eq:additive-noise-channel}, $x$ and $p$ are
zero-mean, independent Gaussian random variables each having variance $\xi\geq0$. Let
$\mathcal{T}_{\xi}$ denote this channel. The channel in
\eqref{eq:additive-noise-channel} is entanglement-breaking when $\xi\geq1$
\cite{Holevo2008}, and by Remark~\ref{rem:ent-break}, the entanglement cost is
equal to zero for these values.

Kraus representations for the channels in
\eqref{eq:thermal-channel}--\eqref{eq:additive-noise-channel}\ are available
in \cite{ISS11}, which can be helpful for further understanding their action
on input quantum states.

Due to the entanglement-breaking regions discussed above, we are left with a
limited range of single-mode bosonic Gaussian channels to consider, which is
delineated by the white strip in Figure~1 of \cite{GPCH13}.

\subsection{Upper bound on the entanglement cost of teleportation-simulable
channels with bosonic Gaussian resource states}

In this section, I determine an upper bound on the entanglement cost of any
channel $\mathcal{N}_{A\rightarrow B}$\ that is teleportation simulable with
associated resource state given by a bosonic Gaussian state. Related bosonic
teleportation channels have been considered previously
\cite{prl1998braunstein,BST02,TBS02,GI02,WPG07,NFC09}, in the case that the
LOCC\ channel associated to $\mathcal{N}_{A\rightarrow B}$ is a Gaussian
LOCC\ channel. Proposition~\ref{prop:upper-bnd-gauss} below states that the
entanglement cost of these channels is bounded from above by the Gaussian
entanglement of formation \cite{WGKWC04}\ of the underlying bosonic Gaussian
resource state, and as such, this proposition represents a counterpart to
Proposition~\ref{prop:ach-TP-sim}. Before stating it, let us note that the
Gaussian entanglement of formation $E_{F}^{g}(A;B)_{\rho}$ of a bipartite
state $\rho_{AB}$ \cite{WGKWC04}\ is given by the same formula as in
\eqref{eq:EoF}, with the exception that the pure states $\psi_{AB}^{x}$ in the
ensemble decomposition are required to be Gaussian. Note that continuous probability
measures are allowed for the decomposition (for an explicit definition, see
\cite[Section~III]{WGKWC04}). Let us note here that the first part of the
proof outlines a procedure for the formation of $n$ approximate copies of a
bipartite state, and even though this kind of protocol has been implicit in
prior literature, I have included explicit steps for clarity. After proving
Proposition~\ref{prop:ach-TP-sim}, I discuss its application to thermal,
amplifier, and additive-noise bosonic Gaussian channels.

\begin{proposition}
\label{prop:upper-bnd-gauss}Let $\mathcal{N}_{A\rightarrow B}$ be a channel
that is teleportation simulable as defined in \eqref{eq:tp-sim}, where the
resource state $\omega_{A^{\prime}B^{\prime}}$ is a bosonic Gaussian state
composed of $k$~modes for system $A^{\prime}$ and $\ell$ modes for system
$B^{\prime}$, with $k,\ell\geq1$. Then the entanglement cost of $\mathcal{N}%
_{A\rightarrow B}$ is never larger than the Gaussian entanglement of formation
of the bosonic Gaussian resource state $\omega_{A^{\prime}B^{\prime}}$:%
\begin{equation}
E_{C}(\mathcal{N})\leq E_{F}^{g}(A^{\prime};B^{\prime})_{\omega}.
\end{equation}

\end{proposition}

\begin{proof}
The main idea of the proof is to first form $n$ approximate copies of the
bosonic Gaussian resource state $\omega_{A^{\prime}B^{\prime}}$, by using
entanglement and LOCC as related to the approach from \cite{BBPS96}, and then
after that, simulate $n$ uses of the channel $\mathcal{N}_{A\rightarrow B}$ by
employing the structure of the channel $\mathcal{N}_{A\rightarrow B}$ from
\eqref{eq:tp-sim}. Indispensable to the proof is the analysis in
\cite[Sections~II and III]{WGKWC04}, where it is shown that every bosonic
Gaussian state can be decomposed as a Gaussian mixture of local displacements
acting on a fixed Gaussian pure state and that such a decomposition is optimal
for the Gaussian entanglement of formation \cite[Proposition~1]{WGKWC04}. The
Gaussian mixture of local displacements can be understood as an LOCC channel
$\mathcal{G}_{A^{\prime}B^{\prime}}$, and let $\psi_{A^{\prime}B^{\prime}%
}^{\omega}$\ denote the aforementioned fixed Gaussian pure state such that
$\mathcal{G}_{A^{\prime}B^{\prime}}(\psi_{A^{\prime}B^{\prime}}^{\omega
})=\omega_{A^{\prime}B^{\prime}}$.

Since $\psi_{A^{\prime}B^{\prime}}^{\omega}$ is Gaussian, the marginal state
$\psi_{B^{\prime}}^{\omega}$ is Gaussian, and thus it has finite entropy
$H(B^{\prime})_{\psi^{\omega}}$, as well as finite entropy variance, i.e.,%
\begin{equation}
V(B^{\prime})_{\psi^{\omega}}\equiv\operatorname{Tr}\{\psi_{B^{\prime}%
}^{\omega}[-\log_{2}\psi_{B^{\prime}}^{\omega}-H(B^{\prime})_{\psi^{\omega}%
}]^{2}\}<\infty,
\end{equation}
the latter statement following from the Williamson decomposition
\cite{W36}\ for Gaussian states as well as the formula for the entropy
variance of a bosonic thermal state \cite{WRG16}. For $\delta>0$, recall that
the entropy-typical projector $\Pi_{B^{\prime n}}^{\delta}$
\cite{OP93,PhysRevA.51.2738}\ of the state $\psi_{B^{\prime}}^{\omega}$ is
defined as the projection onto%
\begin{equation}
\operatorname{span}\{|\xi_{z^{n}}\rangle:\left\vert -n^{-1}\log_{2}(p_{Z^{n}%
}(z^{n}))-H(B^{\prime})_{\psi^{\omega}}\right\vert \leq\delta\},
\end{equation}
where a countable spectral decomposition of $\psi_{B^{\prime}}^{\omega}$ is
given by%
\begin{equation}
\psi_{B^{\prime}}^{\omega}=\sum_{z}p_{Z}(z)|\xi_{z}\rangle\langle\xi_{z}|,
\end{equation}
and%
\begin{align}
|\xi_{z^{n}}\rangle &  \equiv|\xi_{z_{1}}\rangle\otimes\cdots\otimes
|\xi_{z_{n}}\rangle,\\
p_{Z^{n}}(z^{n})  &  \equiv p_{Z}(z_{1})\cdots p_{Z}(z_{n}).
\end{align}
The entropy-typical projector $\Pi_{B^{\prime n}}^{\delta}$ projects onto a
finite-dimensional subspace of $[\psi_{B^{\prime}}^{\omega}]^{\otimes n}$, and
satisfies the conditions $\left[  \Pi_{B^{\prime n}}^{\delta},[\psi
_{B^{\prime}}^{\omega}]^{\otimes n}\right]  =0$ and
\begin{multline}
2^{-n\left[  H(B^{\prime})_{\psi^{\omega}}+\delta\right]  }\Pi_{B^{\prime n}%
}^{\delta}\leq\Pi_{B^{\prime n}}^{\delta}[\psi_{B^{\prime}}^{\omega}]^{\otimes
n}\Pi_{B^{\prime n}}^{\delta}\\
\leq2^{-n\left[  H(B^{\prime})_{\psi^{\omega}}-\delta\right]  }\Pi_{B^{\prime
n}}^{\delta}.
\end{multline}
It then follows that $\operatorname{Tr}\{\Pi_{B^{\prime n}}^{\delta}%
\}\leq2^{n\left[  H(B^{\prime})_{\psi^{\omega}}+\delta\right]  }$.
Furthermore, consider that the entropy-typical projector $\Pi_{B^{\prime n}%
}^{\delta}$\ for the state $[\psi_{B^{\prime}}^{\omega}]^{\otimes n}$
satisfies%
\begin{align}
\operatorname{Tr}\{\left(  I_{A^{\prime n}}\otimes\Pi_{B^{\prime n}}^{\delta
}\right)  [\psi_{A^{\prime}B^{\prime}}^{\omega}]^{\otimes n}\}  &
=\operatorname{Tr}\{\Pi_{B^{\prime n}}^{\delta}[\psi_{B^{\prime}}^{\omega
}]^{\otimes n}\}\notag \\
&  \geq1-\frac{V(B^{\prime})_{\psi^{\omega}}}{\delta^{2}n},
\end{align}
with the inequality following from the definition of the entropy-typical
projector and an application of the Chebshev inequality. By the gentle
measurement lemma \cite{itit1999winter,ON07} (see \cite[Lemma~9.4.1]{W17}\ for
the version employed here), we conclude that%
\begin{equation}
\frac{1}{2}\left\Vert [\psi_{A^{\prime}B^{\prime}}^{\omega}]^{\otimes
n}-\widetilde{\psi}_{A^{\prime n}B^{\prime n}}^{\omega}\right\Vert _{1}%
\leq\sqrt{\frac{V(B^{\prime})_{\psi^{\omega}}}{\delta^{2}n}},
\label{eq:gent-meas-approx}%
\end{equation}
where%
\begin{equation}
\widetilde{\psi}_{A^{\prime n}B^{\prime n}}^{\omega}\equiv\frac{\left(
I_{A^{\prime n}}\otimes\Pi_{B^{\prime n}}^{\delta}\right)  [\psi_{A^{\prime
}B^{\prime}}^{\omega}]^{\otimes n}\left(  I_{A^{\prime n}}\otimes
\Pi_{B^{\prime n}}^{\delta}\right)  }{\operatorname{Tr}\{\left(  I_{A^{\prime
n}}\otimes\Pi_{B^{\prime n}}^{\delta}\right)  [\psi_{A^{\prime}B^{\prime}%
}^{\omega}]^{\otimes n}\}}.
\end{equation}
Observe that the system $B^{\prime n}$ of $\widetilde{\psi}_{A^{\prime
n}B^{\prime n}}^{\omega}$ is supported on a finite-dimensional subspace of
$B^{\prime n}$.

Now, the idea of forming $n$ approximate copies $\psi_{A^{\prime}B^{\prime}%
}^{\omega}$ is then the same as it is in \cite{BBPS96}: Alice prepares the
state $\widetilde{\psi}_{A^{\prime n}B^{\prime n}}^{\omega}$ locally, Alice
and Bob require beforehand a maximally entangled state of Schmidt rank no
larger than $2^{n\left[  H(B^{\prime})_{\psi^{\omega}}+\delta\right]  }$, and
then they perform quantum teleportation \cite{BBC+93}\ to teleport the
$B^{\prime n}$ system to Bob. At this point, they share exactly the state
$\widetilde{\psi}_{A^{\prime n}B^{\prime n}}^{\omega}$, which becomes less and
less distinguishable from $[\psi_{A^{\prime}B^{\prime}}^{\omega}]^{\otimes n}$
as $n$ grows large, due to \eqref{eq:gent-meas-approx}. Now applying the
Gaussian LOCC\ channel $(\mathcal{G}_{A^{\prime}B^{\prime}})^{\otimes n}$, the
data processing inequality to \eqref{eq:gent-meas-approx}, and the fact that
$\mathcal{G}_{A^{\prime}B^{\prime}}(\psi_{A^{\prime}B^{\prime}}^{\omega
})=\omega_{A^{\prime}B^{\prime}}$, we conclude that%
\begin{equation}
\frac{1}{2}\left\Vert \omega_{A^{\prime}B^{\prime}}^{\otimes n}-(\mathcal{G}%
_{A^{\prime}B^{\prime}})^{\otimes n}(\widetilde{\psi}_{A^{\prime n}B^{\prime
n}}^{\omega})\right\Vert _{1}\leq\sqrt{\frac{V(B^{\prime})_{\psi^{\omega}}%
}{\delta^{2}n}}. \label{eq:mixed-state-approx}%
\end{equation}

Thus, to see that $H(B^{\prime})_{\psi^{\omega}}$ is an achievable rate for
forming $\omega_{A^{\prime}B^{\prime}}^{\otimes n}$, fix $\varepsilon\in(0,1]$
and $\delta>0$. Then choose $n$ large enough so that $\sqrt{\frac{V(B^{\prime
})_{\psi^{\omega}}}{\delta^{2}n}}\leq\varepsilon$. Apply the above procedure,
using LOCC\ and a maximally entangled state of Schmidt rank no larger than
$2^{n\left[  H(B^{\prime})_{\psi^{\omega}}+\delta\right]  }$. Then the rate of
entanglement consumption to produce $n$ approximate copies of $\omega
_{A^{\prime}B^{\prime}}$
satisfying \eqref{eq:mixed-state-approx} is $H(B^{\prime})_{\psi^{\omega}}+\delta$. Since
this is possible for $\varepsilon\in(0,1]$, $\delta>0$, and sufficiently large
$n$, we conclude that $H(B^{\prime})_{\psi^{\omega}}$ is an achievable rate
for the formation of $\omega_{A^{\prime}B^{\prime}}$. Now, since achieving
this rate is possible for any pure state $\psi_{A^{\prime}B^{\prime}}^{\omega
}$ such that $\omega_{A^{\prime}B^{\prime}}=\mathcal{G}_{A^{\prime}B^{\prime}%
}(\psi_{A^{\prime}B^{\prime}}^{\omega})$, we conclude that the infimum of
$H(B^{\prime})_{\psi^{\omega}}$ with respect to all such pure states is an
achievable rate. But this latter quantity is exactly the Gaussian entanglement
of formation according to \cite[Proposition~1]{WGKWC04}.

The idea for simulating $n$ uses of the channel $\mathcal{N}_{A\rightarrow B}$
is then the same as the idea used in the proof of
Proposition~\ref{prop:ach-TP-sim}. First form $n$ approximate copies of
$\omega_{A^{\prime}B^{\prime}}$ according to the procedure described above.
Then, when the $i$th call to the channel $\mathcal{N}_{A\rightarrow B}$ is
made, use the LOCC\ channel $\mathcal{L}_{AA^{\prime}B^{\prime}\rightarrow B}$
from the definition in \eqref{eq:tp-sim} along with the $i$th $A^{\prime}$ and
$B^{\prime}$ systems of the state approximating $\omega_{A^{\prime}B^{\prime}%
}^{\otimes n}$ to simulate it. By the same reasoning that led to
\eqref{eq:limitation-on-dist-seq-sim}, the distinguishability of the final
states of any sequential test is limited by the distinguishability of the
state $\omega_{A^{\prime}B^{\prime}}^{\otimes n}$ from its approximation,
which I argued in \eqref{eq:mixed-state-approx}\ can be made arbitrarily small
with increasing $n$. Thus, the Gaussian entanglement of formation
$\omega_{A^{\prime}B^{\prime}}$ is an achievable rate for sequential channel
simulation of $\mathcal{N}_{A\rightarrow B}$.
\end{proof}

\subsubsection{Upper bound for the entanglement cost of thermal, amplifier,
and additive-noise bosonic Gaussian channels}

\label{sec:app-to-gaussian}I now discuss how to apply
Proposition~\ref{prop:ach-TP-sim} to single-mode thermal, amplifier, and
additive-noise channels. Some recent papers \cite{LMGA17,KW17,TDR18}\ have
shown how to simulate each of these channels by using a bosonic Gaussian
resource state along with variations of the continuous-variable quantum
teleportation protocol \cite{prl1998braunstein}. Of these works, the one most
relevant for us is the latest one \cite{TDR18}, because these authors proved
that the entanglement of formation of the underlying resource state is equal
to the entanglement of formation that results from transmitting through the channel one share of a
two-mode squeezed vacuum state with arbitrarily large squeezing strength. That
is, let $\mathcal{N}_{A\rightarrow B}$ denote a single-mode thermal,
amplifier, or additive-noise channel. Then one of the main results of
\cite{TDR18}\ is that, associated to this channel, there is a bosonic Gaussian
resource state $\omega_{A^{\prime}B^{\prime}}$ and a Gaussian LOCC\ channel
$\mathcal{G}_{AA^{\prime}B^{\prime}\rightarrow B}$ such that%
\begin{align}
E_{F}(A^{\prime};B^{\prime})_{\omega}  &  =\sup_{N_{S}\geq0}E_{F}%
(R;B)_{\sigma(N_{S})}\label{eq:ralph-1}\\
&  =\lim_{N_{S}\rightarrow\infty}E_{F}(R;B)_{\sigma(N_{S})},
\label{eq:ralph-2}%
\end{align}
where%
\begin{align}
\sigma(N_{S})  &  \equiv\mathcal{N}_{A\rightarrow B}(\phi_{RA}^{N_{S}}),\\
\phi_{RA}^{N_{S}}  &  \equiv|\phi^{N_{S}}\rangle\langle\phi^{N_{S}}|_{RA},\\
|\phi^{N_{S}}\rangle_{RA}  &  \equiv\frac{1}{\sqrt{N_{S}+1}}\sum_{n=0}%
^{\infty}\sqrt{\left(  \frac{N_{S}}{N_{S}+1}\right)  ^{n}}|n\rangle
_{R}|n\rangle_{A}, \label{eq:TMSV}%
\end{align}
and for all input states$~\rho_{A}$,%
\begin{equation}
\mathcal{N}_{A\rightarrow B}(\rho_{A})=\mathcal{G}_{AA^{\prime}B^{\prime
}\rightarrow B}(\rho_{A}\otimes\omega_{A^{\prime}B^{\prime}}).
\end{equation}
In the above, $\phi_{RA}^{N_{S}}$ is the two-mode squeezed vacuum state
\cite{S17}. Note that the equality in \eqref{eq:ralph-2}\ holds because one
can always produce $\phi_{RA}^{N_{S}}$ from $\phi_{RA}^{N_{S}^{\prime}}$ such
that $N_{S}^{\prime}\geq N_{S}$, by using Gaussian LOCC\ and the local
displacements involved in the Gaussian LOCC commute with the channel
$\mathcal{N}_{A\rightarrow B}$ \cite{GECP03} (whether it be thermal,
amplifier, or additive-noise). Furthermore, the entanglement of formation does
not increase under the action of an LOCC\ channel.

Thus, applying the above observations and
Proposition~\ref{prop:upper-bnd-gauss}, it follows that there exist bosonic
Gaussian resource states $\omega_{A^{\prime}B^{\prime}}^{\eta,N_{B}}$,
$\omega_{A^{\prime}B^{\prime}}^{G,N_{B}}$, and $\omega_{A^{\prime}B^{\prime}%
}^{\xi}$ associated to the respective thermal, amplifier, and additive-noise
channels in \eqref{eq:thermal-channel}--\eqref{eq:additive-noise-channel},
such that the following inequalities hold%
\begin{align}
E_{C}(\mathcal{L}_{\eta,N_{B}})  &  \leq E_{F}(A^{\prime};B^{\prime}%
)_{\omega^{\eta,N_{B}}},\label{eq:e-cost-upper-thermal}\\
E_{C}(\mathcal{A}_{G,N_{B}})  &  \leq E_{F}(A^{\prime};B^{\prime}%
)_{\omega^{G,N_{B}}},\\
E_{C}(\mathcal{T}_{\xi})  &  \leq E_{F}(A^{\prime};B^{\prime})_{\omega^{\xi}}.
\label{eq:e-cost-upper-additive-noise}%
\end{align}
Analytical formulas for the upper bounds on the right can be found in
\cite[Eqs.~(4)--(6)]{TDR18}.

\subsection{Lower bound on the entanglement cost of bosonic Gaussian channels}

In\ this section, I establish a lower bound on the non-asymptotic entanglement
cost of thermal, amplifier, or additive-noise bosonic Gaussian channels. After
that, I show how this bound implies a lower bound on the entanglement cost.
Finally, by proving that the state resulting from sending one share of a
two-mode squeezed vacuum through a pure-loss or pure-amplifier channel has
entanglement cost equal to entanglement of formation, I establish the exact
entanglement cost of these channels by combining with the results from the
previous section.

\begin{proposition}
\label{prop:non-asym-lower-bnd-e-cost-bosonic}Let $\mathcal{N}_{A\rightarrow
B}$ be a thermal, amplifier, or additive-noise channel, as defined in
\eqref{eq:thermal-channel}--\eqref{eq:additive-noise-channel}. Let
$n,M\in\mathbb{N}$, $\varepsilon\in\lbrack0,1/2)$, $\varepsilon^{\prime}%
\in(\sqrt{2\varepsilon},1]$, $\delta=\left[  \varepsilon^{\prime}%
-\sqrt{2\varepsilon}\right]  /\left[  1+\varepsilon^{\prime}\right]  $, and
$N_{S}\in\lbrack0,\infty)$. Then the following bound holds for any
$(n,M,\varepsilon)$ sequential or parallel channel simulation code for $\mathcal{N}%
_{A\rightarrow B}$:%
\begin{multline}
\frac{1}{n}\log_{2}M\geq\frac{1}{n}E_{F}(R^{n};B^{n})_{\omega^{\otimes n}%
}-\left(  \varepsilon^{\prime}+2\delta\right)  H(\phi_{R}^{N_{S}/\delta})\\
-\frac{1}{n}\left[  2\left(  1+\varepsilon^{\prime}\right)  g_{2}%
(\varepsilon^{\prime})+2h_{2}(\delta)\right]  ,
\end{multline}
where $\omega_{RB}\equiv\mathcal{N}_{A\rightarrow B}(\phi_{RA}^{N_{S}})$ and
$\frac{1}{n}\log_{2}M$ is understood as the non-asymptotic entanglement cost
of the protocol.
\end{proposition}

\begin{proof}
The reasoning here is very similar to that given in the proof of
Proposition~\ref{prop:non-asym-lower-bnd-e-cost}, but we can instead make use of
the continuity bound for the entanglement of formation of energy-constrained
states \cite[Proposition~5]{S16}. To begin, suppose that there exists an
$(n,M,\varepsilon)$ protocol for sequential channel simulation. Then by
previous reasoning (also see Figure~\ref{fig:sim-parallel}), it can be thought
of as a parallel channel simulation protocol, such that the criterion in
\eqref{eq:berta-sim}\ holds. Let us take $(\phi_{RA}^{N_{S}})^{\otimes n}$ to
be a test input state, leading to $\omega_{RB}^{\otimes n}=[\mathcal{N}%
_{A\rightarrow B}(\phi_{RA}^{N_{S}})]^{\otimes n}$ when the actual channels
are applied and $\sigma_{R_{1}\cdots R_{n}B_{1}\cdots B_{n}}$ when the
simulation is applied. Set%
\begin{multline}
f(n,\varepsilon,\varepsilon^{\prime},N_{S})\equiv n\left(  \varepsilon
^{\prime}+2\delta\right)  H(\phi_{R}^{N_{S}/\delta})\\
+2\left(  1+\varepsilon^{\prime}\right)  g_{2}(\varepsilon^{\prime}%
)+2h_{2}(\delta).
\end{multline}
Then we have that%
\begin{align}
&  E_{F}(R^{n};B^{n})_{\omega^{\otimes n}}\nonumber\\
&  \leq E_{F}(R^{n};B^{n})_{\sigma}+f(n,\varepsilon,\varepsilon^{\prime}%
,N_{S})\nonumber\\
&  \leq E_{F}(R^{n}A^{n}\overline{A}_{0};\overline{B}_{0})_{\psi\otimes\Phi
}+f(n,\varepsilon,\varepsilon^{\prime},N_{S})\nonumber\\
&  =E_{F}(\overline{A}_{0};\overline{B}_{0})_{\Phi}+f(n,\varepsilon
,\varepsilon^{\prime},N_{S})\nonumber\\
&  =\log_{2}M+f(n,\varepsilon,\varepsilon^{\prime},N_{S}).
\end{align}
The first inequality follows from the condition in \eqref{eq:good-sim}, as
well as from the continuity bound for entanglement of formation from
\cite[Proposition~5]{S16}, noting that the total photon number of the reduced
(thermal) state on systems $R^{n}$ is equal to $nN_{S}$. The second inequality
follows from the LOCC\ monotonicity of the entanglement of formation, here
thinking of the person who possesses systems $RA^{n}$ to be in the same
laboratory as the one possessing the systems $\overline{A}_{i}$, while the
person who possesses the $\overline{B}_{i}$ systems is in a different
laboratory. The first equality follows from the fact that $(\phi_{RA}^{N_{S}%
})^{\otimes n}$ is in tensor product with $\Phi_{\overline{A}_{0}\overline
{B}_{0}}$, so that by a local channel, one may remove $(\phi_{RA}^{N_{S}%
})^{\otimes n}$ or append it for free. The final equality follows because the
entanglement of formation of the maximally entangled state is equal to the
logarithm of its Schmidt rank.
\end{proof}

A direct consequence of
Proposition~\ref{prop:non-asym-lower-bnd-e-cost-bosonic} is the following
lower bound on the entanglement cost of the thermal, amplifier, and
additive-noise channels:

\begin{proposition}
\label{prop:e-cost-lower-bnd-bosonic}Let $\mathcal{N}_{A\rightarrow B}$ be a
thermal, amplifier, or additive-noise channel, as defined in
\eqref{eq:thermal-channel}--\eqref{eq:additive-noise-channel}. Then the
entanglement costs $E_{C}(\mathcal{N})$
and $E_{C}^{(p)}(\mathcal{N})$
are bounded from below by the entanglement
cost of the state $\mathcal{N}_{A\rightarrow B}(\phi_{RA}^{N_{S}})$, where the two-mode
squeezed vacuum state
$\phi_{RA}^{N_{S}}$ has arbitrarily large squeezing strength:%
\begin{align}
E_{C}(\mathcal{N})  & \geq E_{C}^{(p)}(\mathcal{N})
\\
& \geq\sup_{N_{S}\geq0}E_{C}(\mathcal{N}_{A\rightarrow
B}(\phi_{RA}^{N_{S}}))\\
&  = \lim_{N_{S}\to\infty}E_{C}(\mathcal{N}_{A\rightarrow B}(\phi_{RA}^{N_{S}%
})). \label{eq:limit-e-cost-inf-sq}%
\end{align}

\end{proposition}

\begin{proof}
The first inequality follows from definitions, as argued previously in \eqref{eq:sequential-greater-than-para}.
To arrive at the second inequality, in Proposition~\ref{prop:non-asym-lower-bnd-e-cost-bosonic}, set
$\varepsilon^{\prime}=\sqrt[4]{2\varepsilon}$, and take the limit as
$n\rightarrow\infty$ and then as $\varepsilon\rightarrow0$. Employing the fact
that $\lim_{\xi\rightarrow0}\xi H(H(\phi_{R}^{N_{S}/\xi}))=0$
\cite[Proposition~1]{S06} and applying definitions, we find for all $N_{S}%
\geq0$ that%
\begin{align}
E_{C}(\mathcal{N}) &\geq E_{C}^{(p)}(\mathcal{N})\\
 &  \geq\lim_{n\rightarrow\infty}\frac{1}{n}E_{F}%
([\mathcal{N}_{A\rightarrow B}(\phi_{RA}^{N_{S}})]^{\otimes n})\\
&  =E_{C}(\mathcal{N}_{A\rightarrow B}(\phi_{RA}^{N_{S}})).
\end{align}
Since the above bound holds for all $N_{S}\geq0$, we conclude the bound in the
statement of the proposition. The equality in \eqref{eq:limit-e-cost-inf-sq}
follows for the same reason as given for the equality in \eqref{eq:ralph-2},
and due to the fact that entanglement cost is non-increasing with respect to
an LOCC channel by definition.
\end{proof}

\subsection{Additivity of entanglement of formation for pure-loss and
pure-amplifier channels}

The bound in Proposition~\ref{prop:e-cost-lower-bnd-bosonic} is really only a
formal statement, as it is not clear how to evaluate the lower bound
explicitly. If it would however be possible to prove that%
\begin{equation}
\frac{1}{n}E_{F}([\mathcal{N}_{A\rightarrow B}(\phi_{RA}^{N_{S}})]^{\otimes
n})\overset{?}{=}E_{F}(\mathcal{N}_{A\rightarrow B}(\phi_{RA}^{N_{S}}))
\label{eq:additivity-bosonic}%
\end{equation}
for all integer $n\geq1$ and all $N_{S}\geq0$, then we could conclude the
following%
\begin{equation}
E_{C}(\mathcal{N})\overset{?}{\geq}\lim_{N_{S}\to\infty}E_{F}(\mathcal{N}%
_{A\rightarrow B}(\phi_{RA}^{N_{S}})),
\end{equation}
implying that this lower bound coincides with the upper bound from
\eqref{eq:e-cost-upper-thermal}--\eqref{eq:e-cost-upper-additive-noise}, due
to the recent result of \cite{TDR18}\ recalled in \eqref{eq:ralph-1}--\eqref{eq:ralph-2}.

In Proposition~\ref{prop:additivity-EoF-gaussian} below, I prove that the
additivity relation in \eqref{eq:additivity-bosonic} indeed holds whenever the
channel $\mathcal{N}_{A\rightarrow B}$ is a pure-loss channel $\mathcal{L}%
_{\eta}$ or pure-amplifier channel $\mathcal{A}_{G}$. The linchpin of the
proof is the multi-mode bosonic minimum output entropy theorem from
\cite{GHG15} and \cite[Theorem~1]{Giovannetti2015}.

\begin{proposition}
\label{prop:additivity-EoF-gaussian}For $\mathcal{N}_{A\rightarrow B}$ a
pure-loss channel $\mathcal{L}_{\eta}$ with transmissivity $\eta\in(0,1)$ or a
pure-amplifier channel $\mathcal{A}_{G}$ with gain $G>1$, the following
additivity relation holds for all integer $n\geq1$ and $N_{S}\in
\lbrack0,\infty)$:%
\begin{align}
\frac{1}{n}E_{F}([\mathcal{N}_{A\rightarrow B}(\phi_{RA}^{N_{S}})]^{\otimes
n})  &  =E_{F}(\mathcal{N}_{A\rightarrow B}(\phi_{RA}^{N_{S}}%
))\label{eq:additivity-1-eof}\\
&  =E_{F}^{g}(\mathcal{N}_{A\rightarrow B}(\phi_{RA}^{N_{S}})),
\label{eq:additivity-2-eof}%
\end{align}
where $\phi_{RA}^{N_{S}}$ is the two-mode squeezed vacuum state from
\eqref{eq:TMSV} and $E_{F}^{g}$ denotes the Gaussian entanglement of
formation. Thus, the entanglement cost of $\mathcal{N}_{A\rightarrow B}%
(\phi_{RA}^{N_{S}})$ is equal to its entanglement of formation:%
\begin{equation}
E_{C}(\mathcal{N}_{A\rightarrow B}(\phi_{RA}^{N_{S}}))=E_{F}(\mathcal{N}%
_{A\rightarrow B}(\phi_{RA}^{N_{S}})). \label{eq:e-cost-single-letterize}%
\end{equation}

\end{proposition}

\begin{proof}
The proof of this proposition relies on three key prior results:

\begin{enumerate}
\item The main result of \cite{KW04} is that the entanglement of formation
$E_{F}(A;B)_{\psi}$ is equal to the classically-conditioned entropy
$H(A|\overline{E})_{\psi}$ for a tripartite pure state $\psi_{ABE}$:%
\begin{equation}
E_{F}(A;B)_{\psi}=H(A|\overline{E})_{\psi}, \label{eq:kw}%
\end{equation}
where%
\begin{equation}
H(A|\overline{E})_{\psi}=\inf_{\{\Lambda_{E}^{x}\}_{x}}\sum_{x}p_{X}%
(x)H(A)_{\sigma^{x}}, \label{eq:classical-cond-ent}%
\end{equation}
with the optimization taken with respect to a positive operator-valued measure
$\{\Lambda_{E}^{x}\}_{x}$ and%
\begin{align}
p_{X}(x)  &  \equiv\operatorname{Tr}\{\Lambda_{E}^{x}\psi_{E}\},\\
\sigma_{A}^{x}  &  \equiv\frac{1}{p_{X}(x)}\operatorname{Tr}_{E}%
\{(I_{A}\otimes\Lambda_{E}^{x})\psi_{AE}\}.
\end{align}
The sum in \eqref{eq:classical-cond-ent}\ can be replaced with an integral for
continuous-outcome measurements. The equality in \eqref{eq:kw}\ can be
understood as being a consequence of the quantum steering effect
\cite{Schroedinger1935}.

\item The determination of and method of proof for the classically-conditioned
entropy $H(A|\overline{E})_{\rho}$ of an arbitrary two-mode Gaussian state
$\rho_{AE}$ with covariance matrix in certain standard forms \cite{PSBCL14}.
(As remarked below, there is in fact a significant strengthening of the main
result of \cite{PSBCL14}, which relies on item~3 below.)

\item The multi-mode bosonic minimum output entropy theorem from \cite{GHG15}
and \cite[Theorem~1]{Giovannetti2015} (see the related work in
\cite{MGH13,GPCH13} also), which implies that the following identity holds for
a phase-insensitive, single-mode bosonic Gaussian channel $\mathcal{G}$ and
for all integer $n\geq1$:%
\begin{align}
\inf_{\rho^{(n)}}H(\mathcal{G}^{\otimes n}(\rho^{(n)})) & =
H(\mathcal{G}^{\otimes n}%
([|0\rangle\langle0|]^{\otimes n}))
\notag \\
& = nH(\mathcal{G}%
(|0\rangle\langle0|)), \label{eq:multi-mode-min-out-ent}%
\end{align}
where the optimization is with respect to an arbitrary $n$-mode input state
$\rho^{(n)}$ and $|0\rangle\langle0|$ denotes the bosonic vacuum state.
\end{enumerate}

Indeed, these three key ingredients, with the third being the linchpin, lead
to the statement of the proposition after making a few observations. Consider
that a purification of the state $\rho_{AB}=(\operatorname{id}_{R\rightarrow
A}\otimes\mathcal{L}_{\eta})(\phi_{RA}^{N_{S}})$ is given by%
\begin{equation}
\psi_{ABE}=(\operatorname{id}_{R\rightarrow A}\otimes\mathcal{B}%
_{AE\rightarrow BE}^{\eta})(\phi_{RA}^{N_{S}}\otimes|0\rangle\langle0|_{E}),
\end{equation}
where $\mathcal{B}_{AE\rightarrow BE}^{\eta}$ represents the unitary for a
beamsplitter interaction \cite{S17}\ and $|0\rangle\langle0|_{E}$ again
denotes the vacuum state.\ Tracing over the system $B$ gives the state
$\psi_{AE}=(\operatorname{id}_{R\rightarrow A}\otimes\mathcal{L}_{1-\eta
})(\phi_{RA}^{N_{S}})$, where $\mathcal{L}_{1-\eta}$ is a pure-loss channel of
transmissivity $1-\eta$. The state $\psi_{AE}$ is well known to have its
covariance matrix in standard form \cite{S17} (see discussion surrounding
\cite[Eq.~(5)]{PSBCL14}) as%
\begin{equation}%
\begin{bmatrix}
a & 0 & c & 0\\
0 & a & 0 & -c\\
c & 0 & b & 0\\
0 & -c & 0 & b
\end{bmatrix}
\end{equation}
and is also known as a two-mode squeezed thermal state \cite{S17}. As such,
the main result of \cite{PSBCL14} applies, and we can conclude that heterodyne
detection is the optimal measurement in \eqref{eq:classical-cond-ent}, which
in turn implies from \eqref{eq:kw}\ that the entanglement of formation of
$\rho_{AB}$ is equal to the Gaussian entanglement of formation.

However, what we require is that the same results hold for the multi-copy
state $\psi_{AE}^{\otimes n}$. Inspecting Eqs.~(9)--(14) of \cite{PSBCL14}, it
is clear that the same steps hold, except that we replace Eq.~(12) therein
with \eqref{eq:multi-mode-min-out-ent}. Thus, it follows that $n$ individual
heterodyne detections on each $E$ mode of $\psi_{AE}^{\otimes n}$ is the
optimal measurement, so that%
\begin{equation}
\frac{1}{n}H(A^{n}|\overline{E^{n}})_{\psi^{\otimes n}}=H(A|\overline
{E})_{\psi}. \label{eq:single-letterize-cond-ent}%
\end{equation}
By applying \eqref{eq:kw} (as applied to the states $\rho
_{AB}^{\otimes n}$ and $\psi_{AE}^{\otimes n}$), we conclude that%
\begin{equation}
\frac{1}{n}E_{F}(A^{n};B^{n})_{\rho^{\otimes n}}=E_{F}(A;B)_{\rho}.
\end{equation}
Furthermore, since the optimal measurement is given by heterodyne detection,
performing it on mode $E$ of $\psi_{ABE}$ induces a Gaussian ensemble of pure
states $\{p_{X}(x),\psi_{AB}^{x}\}$, which is the optimal decomposition of
$\psi_{AB}=\rho_{AB}$, and thus we conclude that $E_{F}(A;B)_{\rho}=E_{F}%
^{g}(A;B)_{\rho}$.

A similar analysis applies for the quantum-limited amplifier channel. I give
the argument for completeness. Consider that a purification of the state
$\sigma_{AB}=(\operatorname{id}_{R\rightarrow A}\otimes\mathcal{A}_{G}%
)(\phi_{RA}^{N_{S}})$ is given by%
\begin{equation}
\varphi_{ABE}=(\operatorname{id}_{R\rightarrow A}\otimes\mathcal{S}%
_{AE\rightarrow BE}^{G})(\phi_{RA}^{N_{S}}\otimes|0\rangle\langle0|_{E}),
\end{equation}
where $\mathcal{S}_{AE\rightarrow BE}^{G}$ represents the unitary for a
two-mode squeezer \cite{S17}\ and $|0\rangle\langle0|_{E}$ again denotes the
vacuum state.\ Tracing over the system $B$ gives the state $\varphi
_{AE}=(\operatorname{id}_{R\rightarrow A}\otimes\widetilde{\mathcal{A}}%
_{G})(\phi_{RA}^{N_{S}})$, where $\widetilde{\mathcal{A}}_{G}$ denotes the
channel conjugate to the quantum-limited amplifier. The state $\varphi_{AE}$
has its covariance matrix in the form (see Mathematica files included with the
arXiv posting or alternatively \cite[Appendix~D.4]{S17})%
\begin{equation}%
\begin{bmatrix}
a & 0 & c & 0\\
0 & a & 0 & c\\
c & 0 & b & 0\\
0 & c & 0 & b
\end{bmatrix}
,
\end{equation}
and so the same proof approach to get \eqref{eq:single-letterize-cond-ent} can
be used to conclude that%
\begin{equation}
\frac{1}{n}H(A^{n}|\overline{E^{n}})_{\varphi^{\otimes n}}=H(A|\overline
{E})_{\varphi}.
\end{equation}
Indeed, this additionally follows from the discussion after
\cite[Eqs.~(17)--(19)]{PSBCL14}. As such, we conclude in the same way that%
\begin{equation}
\frac{1}{n}E_{F}(A^{n};B^{n})_{\sigma^{\otimes n}}=E_{F}(A;B)_{\sigma}%
=E_{F}^{g}(A;B)_{\sigma}.
\end{equation}

The final statement about entanglement cost in
\eqref{eq:e-cost-single-letterize} follows from the fact that it is equal to the regularized entanglement of formation.
\end{proof}

\begin{remark}
As can be seen from the proof above, the multi-mode minimum output entropy
theorem recalled in \eqref{eq:multi-mode-min-out-ent} provides a significant
strengthening of the results from \cite{PSBCL14}. Indeed, for $\rho_{AE}$ any
two-mode Gaussian state considered in \cite{PSBCL14}, the following equality
holds%
\begin{equation}
\frac{1}{n}H(A^{n}|\overline{E^{n}})_{\rho^{\otimes n}}=H(A|\overline
{E})_{\rho},
\end{equation}
implying that the measurement $\{\Lambda_{E}^{x}\}_{x}$ optimal for the
right-hand side leads to a measurement $\{\Lambda_{E_{1}}^{x_{1}}\otimes
\cdots\otimes\Lambda_{E_{n}}^{x_{n}}\}_{x_{1},\ldots,x_{n}}$ that is optimal
for the left-hand side. Furthermore, by the relation in \eqref{eq:kw}, for any
purification $\psi_{ABE}$ of the state $\rho_{AE}$ mentioned above, we
conclude that%
\begin{equation}
\frac{1}{n}E_{F}(A^{n};B^{n})_{\psi^{\otimes n}}=E_{F}(A;B)_{\psi},
\label{eq:additivity-2-mode-g}%
\end{equation}
for all integer $n\geq1$, thus giving a whole host of two-mode Gaussian states
for which their entanglement cost is equal to their entanglement of formation:
$E_{F}(A;B)_{\rho}=E_{C}(\rho_{AB})=E_{F}^{g}(A;B)_{\rho}$. As far as I am
aware, these are the first examples of two-mode Gaussian states for which the
additivity relation in \eqref{eq:additivity-2-mode-g}\ has been explicitly shown.
\end{remark}

\begin{remark}
One might wonder whether the same method of proof as given in
Proposition~\ref{prop:additivity-EoF-gaussian} could be used to establish the
equalities in \eqref{eq:additivity-1-eof}--\eqref{eq:additivity-2-eof} for
general thermal, amplifier, and additive-noise channels. At the moment, it is
not clear how to do so. The issue is that the state $(\operatorname{id}%
_{R}\otimes\mathcal{L}_{\eta,N_{B}})(\phi_{RA}^{N_{S}})$ for $N_{B}>0$ is a
faithful state, meaning that it is positive definite and thus has two
symplectic eigenvalues $>1$. This means that any purification of it requires
at least four modes \cite[Section~III-D]{HW01}. Then tracing over the $B$
system leaves a three-mode state, of which we should be measuring two of them,
and so it is not clear how to apply the methods of \cite{PSBCL14} to such a
state. The same issues apply to the states $(\operatorname{id}_{R}%
\otimes\mathcal{A}_{G,N_{B}})(\phi_{RA}^{N_{S}})$ for $N_{B}>0$ and
$(\operatorname{id}_{R}\otimes\mathcal{T}_{\xi})(\phi_{RA}^{N_{S}})$ for
$\xi>0$, which are the states resulting from the amplifier and additive-noise
channels, respectively.
\end{remark}

\subsection{Entanglement cost of pure-loss and pure-amplifier channels}

\label{sec:pure-BGCs}

Based on the results in the previous sections, we conclude the following
theorem, which gives simple formulas for the entanglement cost of two
fundamental bosonic Gaussian channels:

\begin{theorem}
\label{thm:pure-bosonic-formulas} For a pure-loss channel $\mathcal{L}_{\eta}$
with transmissivity $\eta\in(0,1)$ or a pure-amplifier channel $\mathcal{A}%
_{G}$ with gain $G>1$, the following formulas characterize the entanglement
costs of these channels:%
\begin{align}
E_{C}(\mathcal{L}_{\eta})  &  = 
E_{C}^{(p)}(\mathcal{L}_{\eta})
=\frac{h_{2}(1-\eta)}{1-\eta}%
,\label{eq:loss-formula}\\
E_{C}(\mathcal{A}_{G})  &  =
E^{(p)}_{C}(\mathcal{A}_{G}) =
\frac{g_{2}(G-1)}{G-1}, \label{eq:amp-formula}%
\end{align}
where $h_{2}(\cdot)$ is the binary entropy defined in \eqref{eq:bin-ent} and
$g_{2}(\cdot)$ is the bosonic entropy function defined in \eqref{eq:bosonic-ent}.
\end{theorem}

\begin{proof}
Recalling the discussion in Section~\ref{sec:app-to-gaussian}, for a pure-loss
and pure-amplifier channel, there exist respective resource states
$\omega_{A^{\prime}B^{\prime}}^{\eta}$ and $\omega_{A^{\prime}B^{\prime}}^{G}$
such that%
\begin{align}
E_{C}(\mathcal{L}_{\eta})  &  \leq E_{F}(A^{\prime};B^{\prime})_{\omega^{\eta
}}\\
&  =\lim_{N_{S}\rightarrow\infty}E_{F}(R;B)_{\sigma^{\eta}(N_{S}%
)},\label{eq:ralph-result-1}\\
E_{C}(\mathcal{A}_{G})  &  \leq E_{F}(A^{\prime};B^{\prime})_{\omega^{G}}\\
&  =\lim_{N_{S}\rightarrow\infty}E_{F}(R;B)_{\sigma^{G}(N_{S})},
\label{eq:ralph-result-2}%
\end{align}
where%
\begin{align}
\sigma^{\eta}(N_{S})_{RB}  &  \equiv(\operatorname{id}_{R}\otimes
\mathcal{L}_{\eta})(\phi_{RA}^{N_{S}}),\\
\sigma^{G}(N_{S})_{RB}  &  \equiv(\operatorname{id}_{R}\otimes\mathcal{A}%
_{G})(\phi_{RA}^{N_{S}}),
\end{align}
with the equalities in \eqref{eq:ralph-result-1} and
\eqref{eq:ralph-result-2}\ being one of the main results of \cite{TDR18}.
Furthermore, explicit formulas for $E_{F}(A^{\prime};B^{\prime})_{\omega
^{\eta}}$ and $E_{F}(A^{\prime};B^{\prime})_{\omega^{G}}$ have been given in
\cite[Eqs.~(4)--(6)]{TDR18}, and evaluating these formulas leads to the
expressions in \eqref{eq:loss-formula}--\eqref{eq:amp-formula} (supplementary
Mathematica files that automate these calculations are available with the
arXiv posting of this paper).

On the other hand, Propositions~\ref{prop:e-cost-lower-bnd-bosonic} and
\ref{prop:additivity-EoF-gaussian}\ imply that%
\begin{align}
E_{C}(\mathcal{L}_{\eta})  & \geq 
E_{C}^{(p)}(\mathcal{L}_{\eta}) \\& \geq\lim_{N_{S}\rightarrow\infty}E_{C}%
(\sigma^{\eta}(N_{S})_{RB})\\
&  =\lim_{N_{S}\rightarrow\infty}E_{F}(R;B)_{\sigma^{\eta}(N_{S})},\\
E_{C}(\mathcal{A}_{G})  &  \geq
E^{(p)}_{C}(\mathcal{A}_{G}) \\
& \geq 
\lim_{N_{S}\rightarrow\infty}E_{C}(\sigma
^{G}(N_{S})_{RB})\\
&  =\lim_{N_{S}\rightarrow\infty}E_{F}(R;B)_{\sigma^{G}(N_{S})}.
\end{align}
Combining the inequalities above, we conclude the statement of the theorem.
\end{proof}

\bigskip
It is interesting to consider various limits of the formulas in \eqref{eq:loss-formula}--\eqref{eq:amp-formula}:%
\begin{align}
\lim_{\eta\rightarrow1}\frac{h_{2}(1-\eta)}{1-\eta}  &  =\lim_{G\rightarrow
1}\frac{g_{2}(G-1)}{G-1}=\infty,\\
\lim_{\eta\rightarrow0}\frac{h_{2}(1-\eta)}{1-\eta}  &  =\lim_{G\rightarrow
\infty}\frac{g_{2}(G-1)}{G-1}=0.
\end{align}
We expect these to hold because the channels approach the ideal channel in the
limits $\eta,G\rightarrow1$, which we previously argued has infinite
entanglement cost, while they both approach the completely depolarizing
(useless) channel in the no-transmission limit $\eta\rightarrow0$ and
infinite-amplification limit $G\rightarrow\infty$. Furthermore, these formulas
obey the symmetry%
\begin{equation}
\frac{h_{2}(1-\eta)}{1-\eta}=\frac{g_{2}(1/\eta-1)}{1/\eta-1},
\end{equation}
which is consistent with the idea that the transformation $\eta\rightarrow
1/\eta$ takes a channel of transmissivity $\eta\in\lbrack0,1]$ and produces a
channel of gain $1/\eta$. Finally, we have the Taylor expansions:%
\begin{align}
\frac{h_{2}(1-\eta)}{1-\eta}  &  =\frac{\eta}{\ln2}(1-\ln(\eta))+O(\eta
^{2}),\\
\frac{g_{2}(G-1)}{G-1}  &  =\frac{1+\ln(G)}{G\ln2}+O(1/G^{2}),
\end{align}
which are relevant in the low-transmissivity and high-gain regimes.

\begin{figure}[ptb]
\begin{center}
\includegraphics[width=3.4in]
{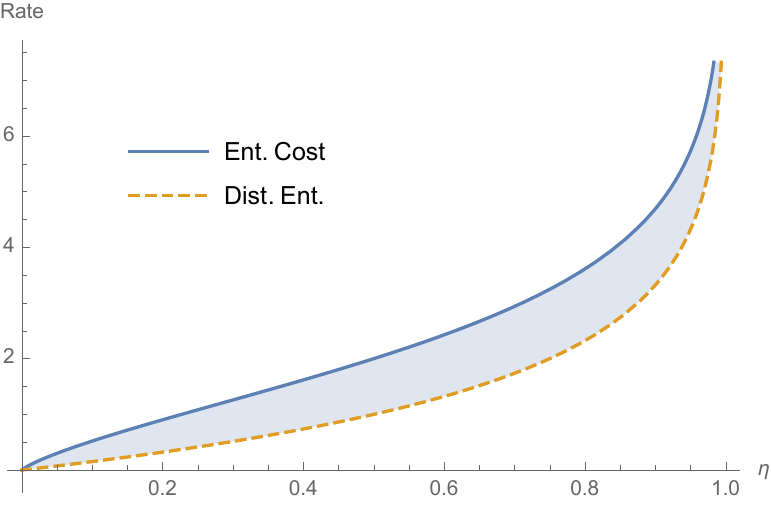}
\end{center}
\caption{Plot of the entanglement cost $E_{C}(\mathcal{L}_{\eta}) =\frac
{h_{2}(1-\eta)}{1-\eta}$ and the distillable entanglement $E_{D}%
(\mathcal{L}_{\eta}) =-\log_{2}(1-\eta)$ of the pure-loss channel
$\mathcal{L}_{\eta}$ as a function of the transmissivity $\eta\in[0,1]$, with
the shaded area demonstrating the gap between them. The units for rate on the vertical axis are ebits per channel use, and $\eta$ on the horizontal axis is dimensionless.}%
\label{fig:EoF-pure-loss}%
\end{figure}

In \cite{PLOB15}, simple formulas for the distillable entanglement of these
channels were determined and given by%
\begin{align}
E_{D}(\mathcal{L}_{\eta})  &  =-\log_{2}(1-\eta),\\
E_{D}(\mathcal{A}_{G})  &  =-\log_{2}(1-1/G).
\end{align}
Thus, the prior results and the formulas in
Theorem~\ref{thm:pure-bosonic-formulas} demonstrate that the resource theory
of entanglement for these channels is irreversible. That is, if one started
from a pure-loss channel of transmissivity $\eta$ and distilled ebits from it
at the ideal rate of $-\log_{2}(1-\eta)$, and then subsequently wanted to use
these ebits to simulate a pure-loss channel with the same transmissivity, this
is not possible, because the rate at which ebits are distilled is not
sufficient to simulate the channel again. The same statement applies to the
pure-amplifier channel. Figures~\ref{fig:EoF-pure-loss} and
\ref{fig:EoF-pure-amp} compare the formulas for entanglement cost and
distillable entanglement of these channels, demonstrating that there is a
noticeable gap between them. I note here that the differences are given by%
\begin{align}
E_{C}(\mathcal{L}_{\eta})-E_{D}(\mathcal{L}_{\eta})  &  =\frac{-\eta\log
_{2}\eta}{1-\eta},\\
E_{C}(\mathcal{A}_{G})-E_{D}(\mathcal{A}_{G})  &  =\frac{\log_{2}G}{G-1},
\end{align}
implying that these differences are strictly greater than zero for all the
relevant channel parameter values $\eta\in(0,1)$ and $G>1$.

\begin{figure}[ptb]
\begin{center}
\includegraphics[width=3.4in]
{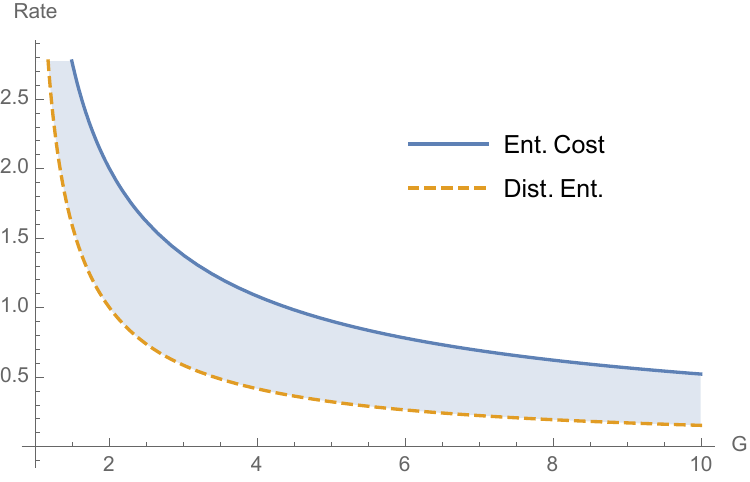}
\end{center}
\caption{Plot of the entanglement cost $E_{C}(\mathcal{A}_{G})=\frac
{g_{2}(G-1)}{G-1}$ and the distillable entanglement $E_{D}(\mathcal{A}%
_{G})=-\log_{2}(1-1/G)$ of the pure-amplifier channel $\mathcal{A}_{G}$ as a
function of the gain $G\geq1$, with the shaded area demonstrating the gap
between them. The units for rate on the vertical axis are ebits per channel use, and $G$ on the horizontal axis is dimensionless.}%
\label{fig:EoF-pure-amp}%
\end{figure}

\section{Extension to other resource theories}

\label{sec:resource-theory-gen}

Let us now consider how to extend many of the
concepts in this paper to other resource theories (see \cite{CG18} for a review of quantum resource theories). In fact, this can be
accomplished on a simple conceptual level by replacing \textquotedblleft
LOCC\ channel\textquotedblright\ with \textquotedblleft free
channel,\textquotedblright\ \textquotedblleft separable
state\textquotedblright\ with \textquotedblleft free state,\textquotedblright%
\ and (roughly) \textquotedblleft maximally entangled state\textquotedblright%
\ with resource state throughout the paper. To be precise, let $F$ denote the
set of free states for a given resource theory, and let $\mathcal{F}$ be a
free channel, which takes a free state to a free state. In \cite[Section~7]%
{KW17a}, a general notion of distillation of a resource from $n$ uses of a
channel was given (see Figure~4 therein). In particular, one interleaves $n$
uses of a given channel by free channels, and the goal is to distill resource
from the $n$ channels. As a generalization of a teleportation-simulable
channel with an associated resource state, the notion of a $\nu$-freely simulable
channel was introduced as a channel $\mathcal{N}$ that can be simulated as%
\begin{equation}
\mathcal{N}_{A\rightarrow B}(\rho_{A})=\mathcal{F}_{AE\rightarrow
B}^{\text{sim}}(\rho_{A}\otimes\nu_{E}),
\end{equation}
where $\mathcal{F}^{\text{sim}}$ is a free channel and $\nu$ is some resource
state. The implications of this for distillation protocols was discussed in
\cite[Section~7]{KW17a}, which is merely that the rate at which resource can
be distilled is limited by the resourcefulness of the underlying resource
state $\nu$.

Going forward, we can also consider a resource-seizable channel in a general resource theory to be a $\nu$-freely simulable
channel for
which, by pre- and post-processing, one can seize the underlying resource
state $\nu$ as%
\begin{equation}
\mathcal{F}_{RB\rightarrow E}^{\text{post}}(\mathcal{N}_{A\rightarrow
B}(\kappa_{RA}^{\text{pre}}))=\nu_{E},
\end{equation}
where $\kappa_{RA}^{\text{pre}}$ is a free state and $\mathcal{F}%
_{RB\rightarrow E}^{\text{post}}$ is a free channel, extending Definition~\ref{def:res-seiz-ch}.

The general notion of channel simulation, as presented in
Section~\ref{sec:ent-cost-new-def}, can be considered in any resource theory
also. Again, the main idea is really to replace \textquotedblleft LOCC
channel\textquotedblright\ with \textquotedblleft free
channel\textquotedblright\ and \textquotedblleft maximally entangled
state\textquotedblright\ with \textquotedblleft resourceful
state\textquotedblright\ in the protocol depicted in
Figure~\ref{fig:adaptive-prot}, and the goal is to determine the minimum rate
at which resourcefulness is needed in order to simulate $n$ uses of a given
channel. If the channels are resource-seizable as discussed above, then the
theory should significantly simplify, as has occurred in this paper for the
entanglement theory of channels (see Theorem~\ref{thm:tp-sim}). Furthermore,
along the lines of the discussion in Section~\ref{sec:LOCC-mon-cost} (and
related to \cite[Section~III-D-5]{CG18}), suppose that a channel
$\mathcal{N}_{A\rightarrow B}$ can be realized from another channel
$\mathcal{M}_{A^{\prime}\rightarrow B^{\prime}}$ via a preprocessing free
channel $\mathcal{F}_{A\rightarrow A^{\prime}M}^{\text{pre}}$ and a
postprocessing free channel $\mathcal{F}_{B^{\prime}M\rightarrow
B}^{\text{post}}$\ as%
\begin{equation}
\mathcal{N}_{A\rightarrow B}=\mathcal{F}_{B^{\prime}M\rightarrow
B}^{\text{post}}\circ\mathcal{M}_{A^{\prime}\rightarrow B^{\prime}}%
\circ\mathcal{F}_{A\rightarrow A^{\prime}M}^{\text{pre}}.
\end{equation}
Then for the same reasons given there, the simulation cost of $\mathcal{N}$
should never exceed the simulation cost of $\mathcal{M}$.

Finally, let us note that some discussions about channel simulation for the
resource theory of coherence have appeared in the last paragraph of
\cite{BGMW17}, as well as the last paragraphs of \cite{DFWRSCW18}. It is clear
from the findings of the present paper that identifying interesting
resource-seizable channels could be a useful first step for understanding
interconversion costs of simulating one channel from another in the resource
theory of coherence. It could also be helpful in further understanding channel
simulation in the resource theory of thermodynamics \cite{FBB18}.

\section{Conclusion}

\label{sec:concl}In summary, this paper has provided a new definition for the
entanglement cost of a channel, in terms of the most general strategy that a
discriminator could use to distinguish $n$ uses of the channel from its
simulation. I established an upper bound on the entanglement cost of a
teleportation-simulable channel in terms of the entanglement cost of the
underlying resource state, and I proved that the bound is saturated in the
case that the channel is resource-seizable (Definition~\ref{def:res-seiz-ch}).
I then established single-letter formulas for the entanglement cost of
erasure, dephasing,  three-dimensional Werner--Holevo channels, and epolarizing channels (complements of depolarizing channels), by
leveraging existing results about the entanglement cost of their Choi states.
I finally considered single-mode bosonic Gaussian channels, establishing
bounds on the entanglement cost of the thermal, amplifier, and additive-noise
channels, while giving simple formulas for the entanglement cost of pure-loss
and pure-amplifier channels. By relating to prior work on the distillable
entanglement of these channels, it became clear that the resource theory of
entanglement for quantum channels is irreversible.

Going forward from here, there are many directions to pursue. The
discrimination protocols considered in Section~\ref{sec:ent-cost-new-def}\ do
not impose any realistic energy constraint on the states that can be used in
discriminating the actual $n$ uses of the channel from the simulation. We
could certainly do so by imposing that the average energy of all the states
input to the actual channel or its simulation should be less than a threshold,
and the result is to demand only that the \textit{energy-constrained strategy
norm} (defined naturally as an extension of both the strategy norm
\cite{CDP08a,GDP09,G12}\ and the energy-constrained diamond norm
\cite{Sh17,Win17}) is less than $\varepsilon\in(0,1)$. To be specific, let
$H_{A}$ be a (positive semi-definite) Hamiltonian acting on the input of the
channel $\mathcal{N}_{A\to B}$ and let $E\in[0,\infty)$ be an energy
constraint. Then, demanding that the supremum in \eqref{eq:good-sim} is taken
over all strategies such that
\begin{equation}
\frac{1}{n}\sum_{i=1}^{n} \operatorname{Tr}\{H_{A} \rho_{A_{i}}\},\ \frac
{1}{n}\sum_{i=1}^{n} \operatorname{Tr}\{H_{A} \tau_{A_{i}}\} \leq E,
\end{equation}
the resulting quantity is an energy-constrained strategy norm. With an energy
constraint in place, one would expect that less entanglement is required to
simulate the channel than if there is no constraint at all, and the resulting
entanglement cost would depend on the given energy constraint. For example, Proposition~\ref{prop:non-asym-lower-bnd-e-cost-bosonic} leads to a lower bound on entanglement cost for an energy-constrained sequential simulation, but it remains open to determine if there is a matching upper bound.

Similar to how measures like squashed entanglement \cite{CW04} and relative
entropy of entanglement \cite{VP98} allow for obtaining converse bounds or
fundamental limitations on the distillation rates of quantum states or
channels, simply by making a clever choice of a squashing channel or separable
state, it would be useful to have a measure like this for bounding
entanglement cost from below. That is, it would be desirable for the measure
to involve a supremum over a given set of test states or channels rather than
an infimum as is the case for squashed entanglement and relative entropy of
entanglement. For example, it would be useful to be able to bound the
entanglement cost of thermal, amplifier, and additive-noise channels from
below, in order to determine how tight are the upper bounds in
\eqref{eq:e-cost-upper-thermal}--\eqref{eq:e-cost-upper-additive-noise}.
Progress on this front is available in \cite{WD17}, but more results in this
area would be beneficial.

One of the main tools used in the analysis of the (parallel) entanglement cost
of channels from \cite{BBCW13}\ is a de-Finetti-style approach, consisting of
the post-selection technique \cite{CKR09}. In particular, the problem of
asymptotic (parallel)\ channel simulation was reduced to simulating the
channel on a single state, called the universal de Finetti state. For the
asymptotic theory of (sequential)\ entanglement cost of channels, could there
be a single universal adaptive channel discrimination protocol to consider,
such that simulating the channel well for such a protocol would imply that it
has been simulated well for all protocols?

For the task of entanglement cost, one could modify the set of free channels
to be either those that completely preserve the positivity of partial
transpose \cite{Rai99,Rai01} or are $k$-extendible in the sense of
\cite{KDWW18}. Could we find simpler lower bounds on entanglement cost of
channels in this way? The semi-definite programming quantity from \cite{WD17}
could be helpful here also. After the above question was posed in the arXiv posting of the present paper, the exact entanglement cost has been solved in \cite{WW18} for the case of \textit{exact} channel simulation, with the set of free channels taken to be those that completely preserve the positivity of partial
transpose. 

Another way to think about quantum channel simulation is to allow the
entanglement to be free but count the cost of classical communication. This
was the approach taken for the reverse Shannon theorem \cite{BDHSW09,BCR09},
and these works also considered only parallel channel simulation. How are the
results there affected if the goal is sequential channel simulation instead?
Is the previous answer from \cite{BDHSW09,BCR09}, the mutual information of
the channel, robust under this change? How do prior results on simulation of
quantum measurements \cite{Winter01a,WHBH12,BRW14} hold up under this change?
A comprehensive summary of results on parallel simulation of quantum channels, including the quantum reverse Shannon theorem, measurement simulation, and entanglement cost,
is available in \cite{B13}.

Finally, is there an example of a channel for which its sequential
entanglement cost is strictly greater than its parallel entanglement cost? The
examples discussed here are those for which either there are equalities or no
conclusion could be drawn. Evidence from quantum channel discrimination
\cite{Harrow10} and related evidence from \cite{DW17} suggests the
possibility. One concrete example to examine in this context is the channel
presented in \cite[Appendix~A]{CM17}, given that it is not implementable from
its image, as discussed there.

\begin{acknowledgements}
I thank Gerardo Adesso, Eneet Kaur, Debbie Leung, Alexander M\"uller-Hermes, Kaushik Seshadreesan, Xin Wang, John Watrous, and Andreas Winter
for helpful discussions, especially Andreas Winter for inspiring discussions
from April 2017 regarding a general notion of channel simulation. I
acknowledge support from the National Science Foundation under grant no.~1350397.
\end{acknowledgements}

\appendix

\section{Relation between resource-seizable channels and those that are
implementable from their image}

\label{app:res-seize-impl-image}

Definition~\ref{def:res-seiz-ch}\ introduced the notion of a resource-seizable
channel, and Section~\ref{sec:resource-theory-gen}\ discussed how this notion
can play a role in an arbitrary resource theory. In \cite[Appendix~A]{CM17}, a
channel $\mathcal{N}_{A\rightarrow B}$\ was defined to be implementable from
its image if there exists a state $\sigma_{A'A}$ and an LOCC\ channel
$\mathcal{L}_{AA'B^{\prime}\rightarrow B}$\ such that the following equality
holds for all input states $\rho_{A}$:%
\begin{equation}
\mathcal{N}_{A\rightarrow B}(\rho_{A})=\mathcal{L}_{AA'B^{\prime}\rightarrow
B}(\rho_{A}\otimes\mathcal{N}_{A^{\prime\prime}\rightarrow B^{\prime}}%
(\sigma_{A^{\prime}A^{\prime\prime}})),\label{eq:imp-from-image}%
\end{equation}
where system  $A^{\prime\prime}$ is isomorphic to system $A$
and system $B^{\prime}$ is isomorphic to system $B$. An example of a channel
that is not implementable from its image was discussed at length in
\cite[Appendix~A]{CM17}.

Here, I prove that a channel is resource-seizable in the resource theory of entanglement if and only if it is
implementable from its image. To see this, suppose that a channel is
implementable from its image. Then, given the above structure in
\eqref{eq:imp-from-image}, it is clear that $\mathcal{N}_{A\rightarrow B}$ is
teleportation simulable with associated resource state given by $\omega
_{A^{\prime}B^{\prime}}=\mathcal{N}_{A^{\prime\prime}\rightarrow B^{\prime}%
}(\sigma_{A^{\prime}A^{\prime\prime}})$. Thus, one can trivially seize the
resource state $\omega_{A^{\prime}B^{\prime}}$ by sending in the input state
$\sigma_{A^{\prime}A^{\prime\prime}}$, which is clearly separable between
Alice and Bob, given that Bob's \textquotedblleft system\textquotedblright%
\ here is trivial.

Now suppose that a teleportation-simulable channel is resource-seizable, as
in Definition~\ref{def:res-seiz-ch}. This means that%
\begin{equation}
\mathcal{N}_{A\rightarrow B}(\rho_{A})=\mathcal{M}_{AA^{\prime}B^{\prime
}\rightarrow B}(\rho_{A}\otimes\omega_{A^{\prime}B^{\prime}}),
\end{equation}
where $\omega_{A^{\prime}B^{\prime}}$ is the resource state and $\mathcal{M}%
_{AA^{\prime}B^{\prime}\rightarrow B}$ is an LOCC\ channel. Furthermore, since
it is resource-seizable, this means that there exists a separable state
$\rho_{A_{M}AB_{M}}$ and a postprocessing LOCC\ channel $\mathcal{D}%
_{A_{M}BB_{M}\rightarrow A^{\prime}B^{\prime}}$ such that%
\begin{equation}
\mathcal{D}_{A_{M}BB_{M}\rightarrow A^{\prime}B^{\prime}}(\mathcal{N}%
_{A\rightarrow B}(\rho_{A_{M}AB_{M}}))=\omega_{A^{\prime}B^{\prime}}.
\end{equation}
To see that the channel is implementable from its image, consider that
$\rho_{A_{M}AB_{M}}$ has a decomposition as follows, given that it is
separable:%
\begin{equation}
\sum_{x}p_{X}(x)\psi_{A_{M}A}^{x}\otimes\phi_{B_{M}}^{x},
\end{equation}
for $p_{X}$ a probability distribution and $\{\psi_{A_{M}A}^{x}\}_{x}$ and
$\{\phi_{B_{M}}^{x}\}_{x}$ sets of pure states. Now define the input state
$\sigma_{A_{M}AX_{A}}$ as%
\begin{equation}
\sigma_{A_{M}AX_{A}}\equiv\sum_{x}p_{X}(x)\psi_{A_{M}A}^{x}\otimes
|x\rangle\langle x|_{X_{A}},
\end{equation}
and note that this is the state we can use for implementing the channel's
image. Define the LOCC\ measure-prepare channel $\mathcal{P}_{X_{A}%
\rightarrow B_{M}}$ as%
\begin{equation}
\mathcal{P}_{X_{A}\rightarrow B_{M}}(\cdot)\equiv\sum_{x}\langle x|_{X_{A}%
}(\cdot)|x\rangle_{X_{A}}\ \phi_{B_{M}}^{x},
\end{equation}
which is understood to be implemented via LOCC by measuring Alice's system $X_A$, communicating the outcome $x$ to Bob, who 
then prepares the state $\phi^x_{B_M}$ based on the outcome. We find that
\begin{multline}
(\mathcal{D}_{A_{M}BB_{M}\rightarrow A^{\prime}B^{\prime}}\circ\mathcal{P}%
_{X_{A}\rightarrow B_{M}}\circ\mathcal{N}_{A\rightarrow B})(\sigma
_{A_{M}AX_{A}})\\
=\omega_{A^{\prime}B^{\prime}}.
\end{multline}
We finally conclude that%
\begin{align}
& \mathcal{N}_{A\rightarrow B}(\rho_{A})\nonumber\\
& =\mathcal{M}_{AA^{\prime}B^{\prime}\rightarrow B}(\rho_{A}\otimes
\omega_{A^{\prime}B^{\prime}})\\
& =\mathcal{L}_{AA_{M}X_{A}\bar{B}\rightarrow B}(\rho_{A}\otimes
\mathcal{N}_{\bar{A}\rightarrow\bar{B}}(\sigma_{A_{M}\bar{A}X_{A}})),
\end{align}
where%
\begin{multline}
\mathcal{L}_{A A_{M}X_{A}\bar{B}\rightarrow B}\equiv\\
\mathcal{M}_{AA^{\prime}B^{\prime}\rightarrow B}\circ\mathcal{D}_{A_{M}\bar
{B}B_{M}\rightarrow A^{\prime}B^{\prime}}\circ\mathcal{P}_{X_{A}\rightarrow
B_{M}},
\end{multline}
so that the channel is implementable from its image by inputting the state
$\sigma_{A_{M}AX_{A}}$ and postprocessing with the LOCC\ channel
$\mathcal{M}_{AA^{\prime}B^{\prime}\rightarrow B}\circ\mathcal{D}_{A_{M}%
BB_{M}\rightarrow A^{\prime}B^{\prime}}\circ\mathcal{P}_{X_{A}\rightarrow
B_{M}}$.

\section{Relation between Choi state of a complementary channel and maximally mixed state sent through isometric extension}

\label{sec:relation-Choi-comp-max-mixed-iso}

The purpose of this appendix is to prove the equality in \eqref{eq:max-mixed-iso-epolar}.
Consider a $d$-dimensional depolarizing channel%
\begin{equation}
\rho\rightarrow\left(  1-p\right)  \rho+p\frac{I}{d}.
\end{equation}
As noted in \cite[Eq.~(3.2)]{DFH06}, a Kraus representation for this channel is as follows:%
\begin{equation}
\left\{  \sqrt{1-p}I,\{\sqrt{p/d}|i\rangle\langle j|\}_{i,j}\right\}
.\label{eq:kraus-op-depo}%
\end{equation}
This is because%
\begin{align}
& \left[  \sqrt{1-p}I\right]  \rho\left[  \sqrt{1-p}I\right] \notag  \\
& \qquad +\sum
_{i,j}\left[  \sqrt{p/d}|i\rangle\langle j|\right]  \rho\left[  \sqrt
{p/d}|j\rangle\langle i|\right] \notag  \\
& =\left(  1-p\right)  \rho+\frac{p}{d}\sum_{i}|i\rangle\langle i|\sum
_{j}\langle j|\rho|j\rangle\\
& =\left(  1-p\right)  \rho+p\operatorname{Tr}\{\rho\}\frac{I}{d}.
\end{align}

Now consider a generic channel $\mathcal{N}_{A\rightarrow B}$\ with Kraus
operators $\{N^{i}\}_{i}$ so that an isometric extension is given by $\sum
_{i}N^{i}\otimes|i\rangle_{E}$. Send the maximally mixed state $\pi=I/d$
through the isometric extension  $\sum_{i}N^{i}\otimes|i\rangle
_{E}$. This leads to the state
\begin{equation}
\frac{1}{d}\sum_{i,j}N^{i}N^{j\dag}\otimes|i\rangle\langle j|_{E}.
\end{equation}
Furthermore, a complementary channel of the original channel, resulting from the isometric extension $\left.\sum
_{i}N^{i}\otimes|i\rangle_{E}\right.$, is then%
\begin{equation}
\rho\rightarrow\mathcal{N}_{A\rightarrow E}^{c}(\rho)=\sum_{i,j}%
\operatorname{Tr}\{N^{i}\rho N^{j\dag}\}|i\rangle\langle j|_{E}.
\end{equation}
The Choi state for this complementary channel is given by%
\begin{align}
\mathcal{N}_{A\rightarrow E}^{c}(\Phi_{RA})  & =\frac{1}{d}\sum_{k,l,i,j}%
|k\rangle\langle l|_{R}\otimes\operatorname{Tr}\{N^{i}|k\rangle\langle
l|_{A}N^{j\dag}\}|i\rangle\langle j|_{E}\notag \\
& =\frac{1}{d}\sum_{k,l,i,j}|k\rangle\langle l|_{R}\otimes\langle
l|_{A}N^{j\dag}N^{i}|k\rangle|i\rangle\langle j|_{E} \notag \\
& =\frac{1}{d}\sum_{k,l,i,j}|k\rangle\langle l|_{A}N^{j\dag}N^{i}%
|k\rangle\langle l|_{R}\otimes|i\rangle\langle j|_{E}\notag \\
& =\frac{1}{d}\sum_{i,j}T(N^{j\dag}N^{i})\otimes|i\rangle\langle j|_{E},
\end{align}
where $T(N^{j\dag}N^{i})$ denotes the transpose of $N^{j\dag}N^{i}$. If it
holds that $N^{i}N^{j\dag}=T(N^{j\dag}N^{i})$, then we conclude that the state
resulting from sending in the maximally mixed state to the isometric extension
of the channel is the same as the Choi state of the complementary channel.
This is the case for the depolarizing channel with the Kraus operators in
\eqref{eq:kraus-op-depo}. Since all complementary channels and isometric
extensions of a channel are related by an isometry acting on the environment
system, we are arrive at the same conclusion for any isometric extension and
the corresponding complementary channel to which it leads.

\section{Matlab code for computing Rains relative entropy}

\label{sec:matlab-listing}

This appendix provides a brief listing of Matlab code that can be used to compute the Rains relative entropy of a bipartite state $\rho_{AB}$ \cite{Rai01,AdMVW02}. The code requires the QuantInf package in order to generate a random state \cite{CubittMatlab}, the CVX package for semi-definite programming optimization
\cite{cvx}, and the CVXQuad package \cite{HFawzi} for relative entropy optimization \cite{Fawzi2018,FF18}.

\begin{widetext}
\begin{lstlisting}[caption=Matlab code for calculating the Rains relative entropy of a random bipartite state $\rho_{AB}$.,label=Rains-rel-ent-matlab]
na = 2; nb = 2;
rho = randRho(na*nb);  % Generate a random bipartite state rho

cvx_begin sdp
    variable tau(na*nb,na*nb) hermitian ;
    minimize ( quantum_rel_entr(rho, tau)/ log(2) );
    tau >= 0;
    norm_nuc(Tx(tau, 2, [ na nb ]))  <= 1;
cvx_end
    
rains_rel_ent = cvx_optval;
\end{lstlisting}
\end{widetext}

\bibliographystyle{alpha}
\bibliography{Ref}

\end{document}